\documentclass{lmcs}
\pdfoutput=1

\usepackage{booktabs}
\newcommand{\textscrb}[1]{%
  \text{\usefont{LS1}{stixscr}{m}{n}#1}%
}

\usepackage{lastpage}
\lmcsdoi{18}{1}{19}
\lmcsheading{}{\pageref{LastPage}}{}{}%
{May~17,~2021}{Jan.~20,~2022}{}

\usepackage[utf8]{inputenc}

\usepackage[table]{}
\usepackage{bm}
\usepackage{listings,amsmath,amssymb,amsthm,amsfonts,tikz}
\usepackage{stmaryrd}
\usetikzlibrary{petri, positioning,automata}
\usepackage{float}
\usepackage{xcolor}
\usepackage[normalem]{ulem}

\keywords{FIFO machines, reachability, underapproximation, counter machines}

\DeclareFontEncoding{LS1}{}{}
\DeclareFontSubstitution{LS1}{stix}{m}{n}
\DeclareMathAlphabet{\mymathbb}{U}{BOONDOX-ds}{m}{n}

\DeclareMathOperator{\boundedL}{\textscrb{L}}

\makeatletter
\DeclareFontFamily{OMX}{MnSymbolE}{}
\DeclareSymbolFont{MnLargeSymbols}{OMX}{MnSymbolE}{m}{n}
\SetSymbolFont{MnLargeSymbols}{bold}{OMX}{MnSymbolE}{b}{n}
\DeclareFontShape{OMX}{MnSymbolE}{m}{n}{
    <-6>  MnSymbolE5
   <6-7>  MnSymbolE6
   <7-8>  MnSymbolE7
   <8-9>  MnSymbolE8
   <9-10> MnSymbolE9
  <10-12> MnSymbolE10
  <12->   MnSymbolE12
}{}
\DeclareFontShape{OMX}{MnSymbolE}{b}{n}{
    <-6>  MnSymbolE-Bold5
   <6-7>  MnSymbolE-Bold6
   <7-8>  MnSymbolE-Bold7
   <8-9>  MnSymbolE-Bold8
   <9-10> MnSymbolE-Bold9
  <10-12> MnSymbolE-Bold10
  <12->   MnSymbolE-Bold12
}{}

\let\llangle\@undefined
\let\rrangle\@undefined
\DeclareMathDelimiter{\llangle}{\mathopen}%
                     {MnLargeSymbols}{'164}{MnLargeSymbols}{'164}
\DeclareMathDelimiter{\rrangle}{\mathclose}%
                     {MnLargeSymbols}{'171}{MnLargeSymbols}{'171}
\makeatother

\newcommand{\IBounded}{IB }

\newcommand{\veca}{\fifocontents{a}}
\newcommand{\Sigmabot}{\Sigma^\bot}
\newcommand{\abL}[1]{L_{#1}^{\mathsf{last}}}

\newcommand{\pairs}[1]{G(#1)}

\newcommand{\exend}{\hfill \ensuremath{\lhd}}

\newcommand{\sem}[1]{\llbracket {#1} \rrbracket}
\newcommand{\cvalue}[1]{\llangle #1 \rrangle}
\newcommand{\fvalue}[2]{\sem{#1}_{#2}}

\newcommand{\crun}[1]{\llangle #1 \rrangle}
\newcommand{\frun}[1]{\sem{#1}}

\newcommand{\ChContents}{\prod_{\ch \in \Ch{\nch}} \alphabet_\ch^\ast}

\makeatletter
\renewcommand{\xRightarrow}[2][]{\ext@arrow 0359\Rightarrowfill@{#1}{#2}}
\makeatother

\newcommand \fifo{\mathcal{S}}
\newcommand\fifostates{Q}
\newcommand\alphabet{\Sigma}
\newcommand\fifotransition{T}

\newcommand\fifocontents[1]{\mathbf{#1}}

\newcommand\init{q_0}

\newcommand\reachset[1]{\textit{Reach}_{#1}}
\newcommand\NP{\textit{NP}}
\newcommand\letter[1]{\mathsf{#1}}
\newcommand\notletter[1]{\bar{\mathsf{#1}}}

\newcommand\inc[1]{\mathsf{inc}(#1)}
\newcommand\dec[1]{\mathsf{dec}(#1)}

\newcommand\N{\mathbb{N}}

\newcommand\sendL{\mathord{\boundedL_!}}
\newcommand\recL{\mathord{\boundedL_?}}

\newcommand{\TestL}[2]{V_{#1}(\textscrb{R})}

\newcommand{\Good}{\textscrb{G}}
\newcommand{\Rel}{\textscrb{R}}
\newcommand{\Parikhimg}[1]{\textup{Parikh}(#1)}

\newcommand{\bpair}[2]{(#1\hspace{0.1em},\hspace{0.1em}#2)}
\newcommand{\boundedLp}{\mathop{\boundedL\hspace{0.01em}'}}
\newcommand{\boundedLpem}{\mathop{\boundedL_!\hspace{-0.2em}'}}

\newcommand\Nat{\mathbb{N}}

\newcommand{\nch}{m}

\newcommand{\action}{\beta}
\newcommand{\ch}{c}

\newcommand{\chb}{d}
\newcommand{\Cnt}{\mathit{Cnt}}
\renewcommand{\epsilon}{\varepsilon}

\newcommand{\CS}{\mathcal{C}}
\newcommand{\cnt}[2]{x_{(#1,#2)}}

\renewcommand{\fifo}{{M}}

\newcommand{\Alphw}{\mathit{Alph}}
\newcommand{\indexw}[2]{\mathit{i}_{#2}}

\newcommand{\Traces}[1]{\mathit{Traces}(#1)}

\newcommand{\aletter}{\fifocontents{a}}
\newcommand{\contents}{\fifocontents{w}}
\newcommand{\chcontents}{\fifocontents{v}}
\newcommand{\cscontents}{\fifocontents{v}}
\newcommand{\contentsp}[1]{\fifocontents{w}_{#1}}
\newcommand{\cscontentsp}[1]{\fifocontents{v}_{#1}}
\newcommand{\contentspp}[1]{\fifocontents{w}_{#1}'}
\newcommand{\cscontentspp}[1]{\fifocontents{v}_{#1}'}
\newcommand{\initcontents}{\bm{\epsilon}}
\newcommand{\initcntcontents}{\boldsymbol{0}}
\newcommand{\reachsetL}[2]{\mathit{Reach}_{#1}(#2)}

\newcommand{\A}{\mathcal{A}}
\newcommand{\Ch}[1]{\mathit{Ch}}

\newcommand{\sendaction}[2]{\langle #2!#1\rangle}
\newcommand{\recaction}[2]{\langle #2?#1\rangle}

\newcommand{\sendtrans}[4]{(#3,\sendaction{#1}{#2},#4)}
\newcommand{\rectrans}[4]{(#3,\recaction{#1}{#2},#4)}

\newcommand{\Sendtransp}[3]{\xrightarrow{\sendaction{#1}{#2}}_{#3}}
\newcommand{\Rectransp}[3]{\xrightarrow{\recaction{#1}{#2}}_{#3}}

\newcommand{\incdec}{\mathit{op}}

\newcommand{\incdectrans}[4]{(#3,(\incdec(#1),#2),#4)}

\newcommand{\incdectransp}[3]{\xrightarrow{(\incdec(#1),#2)}}

\newcommand{\ts}{\mathcal{T}}
\newcommand{\tsrel}[1]{\mathcal{T}}
\newcommand{\TSof}[1]{\mathcal{T}_{#1}}
\newcommand{\fifoAlpha}[1]{A_{#1}}
\newcommand{\csAlpha}[1]{A_{#1}}

\newcommand{\tsstates}{S}
\newcommand{\tsinit}{\mathit{init}}
\newcommand{\tstrans}{\to}
\newcommand{\tstransp}[1]{\xrightarrow{#1}}

\newcommand{\tsstateb}{t}
\newcommand{\tscomp}[2]{#1_#2}
\newcommand{\Suf}{\mathit{Suf}}
\newcommand{\Inf}{\mathit{Infix}}

\newcommand{\runCS}{\tau}
\newcommand{\runFIFO}{\sigma}

\newcommand{\FIFOaction}{\beta}
\newcommand{\CSaction}{\alpha}

\newcommand{\qdead}{q_{\textsf{end}}}
\newcommand{\clive}{d}

\newcommand{\CSmod}{\mathcal{C'}}
\newcommand{\CSmodp}{\mathcal{C''}}

\newcommand{\OBounded}{OB }
\newcommand{\ILBounded}{ILB-}

\begin{document}

\title[Bounded Reachability Problems are Decidable in FIFO Machines]{Bounded Reachability Problems are Decidable\texorpdfstring{\\}{} in FIFO Machines}

\author{Benedikt Bollig\rsuper{a}}
\author{Alain Finkel\rsuper{a,b}}
\author{Amrita Suresh\rsuper{a}}

\address{Universit\'e Paris-Saclay, ENS Paris-Saclay, CNRS, LMF, 91190, Gif-sur-Yvette, France}
\email{\{benedikt.bollig,alain.finkel,amrita.suresh\}@ens-paris-saclay.fr}{}{}%
\address{Institut Universitaire de France}

\begin{abstract}
The undecidability of basic decision problems for general FIFO machines such as reachability and unboundedness is well-known. In this paper, we provide an underapproximation for the general model by considering only runs that are input-bounded (i.e.\ the sequence of messages sent through a particular channel belongs to a given bounded language). We prove, by reducing this model to a counter machine with restricted zero tests, that the rational-reachability problem (and by extension, control-state reachability, unboundedness, deadlock, etc.) is decidable. This class of machines subsumes input-letter-bounded machines, flat machines, linear FIFO nets, and monogeneous machines, for which some of these problems were already shown to be decidable. These theoretical results can form the foundations to build a tool to verify general FIFO machines based on the analysis of input-bounded machines.
\end{abstract}

\maketitle

%% start the paper here:
\section{Introduction}\label{sec: introduction}

 {\bf Context.} Asynchronous distributed processes communicating using First In First Out (FIFO) channels are being widely used for distributed and concurrent programming, and more recently, for web service choreographies.
 Since systems of processes communicating through (at least two) one-directional FIFO channels, or equivalently, machines having a unique control-structure  with a single FIFO channel (acting as a buffer) simulate Turing machines, most properties,
 such as unboundedness of a channel, are undecidable for such systems
 \cite{DBLP:journals/tcs/VauquelinF80,brand1983communicating,DBLP:journals/tcs/MemmiF85}.
 
\paragraph{Reachability in FIFO machines.}
If one restricts
 to runs
 with $B$-bounded channels
 (the number of messages in every channel does not exceed $B$), then reachability becomes decidable for existentially-bounded and universally-bounded FIFO systems \cite{DBLP:journals/fuin/GenestKM07}.
 When limiting the number of phases, the bounded-context reachability problem is in $2$-EXPTIME, even for recursive FIFO systems \cite{DBLP:conf/tacas/TorreMP08,DBLP:conf/fossacs/HeussnerLMS10}. For non-confluent topology, reachability is in EXPTIME for recursive FIFO systems with $1$-bounded channels \cite{DBLP:conf/fossacs/HeussnerLMS10}.
 The notion of $k$-synchronous computations was introduced in \cite{DBLP:conf/cav/BouajjaniEJQ18}.
 Reachability under this restriction and checking $k$-synchronizability are
 both PSPACE-complete \cite{DBLP:conf/fossacs/GiustoLL20}.
 Reachability is in PTIME in half-duplex systems \cite{DBLP:journals/iandc/CeceF05} with two processes (moreover, the reachability set is recognizable and effectively computable), but the natural extension to three processes leads to undecidability. Lossy FIFO systems  (where the channels can lose
 messages) \cite{DBLP:journals/fmsd/AbdullaCBJ04,AF-DistC-94} have been shown to be well-structured and have a decidable (but non-elementary) reachability problem \cite{DBLP:conf/lics/ChambartS08}.
 In \cite{MadhusudanP11,AiswaryaGK14}, uniform criteria for decidability of
 reachability and model-checking
 questions are established
 for communicating recursive systems whose
 restricted architecture or communication mechanism
 gives rise to behaviours of bounded tree-width.

 \newcommand{\Pref}{\mathit{Pref}}
 
\paragraph{Input-bounded FIFO machines.} 
 Many papers, starting in the 80s until today,
 have studied FIFO machines in which the input-language of a channel (i.e.\ the set of words that
 record the messages entering a channel) is included in the set $\Pref(w_1^*w_2^*\ldots w_n^*)$ of prefixes of a bounded language $w_1^*w_2^*\ldots w_n^*$.
 We call this class of FIFO machines \emph{input-bounded}.
 
 If the \emph{set of letters} that may enter a channel $c$ is reduced to a unique letter $a_c$, then the input-language of $c$ is included in $a_c^*$ and this subclass trivially reduces to VASS and Petri nets  \cite{DBLP:journals/ipl/YuG83}.
 Also note that, in general, the behaviour of those FIFO machines does not have bounded tree-width.
 \emph{Monogeneous} FIFO nets \cite{F-Petri82,DBLP:journals/tcs/MemmiF85,DBLP:journals/tcs/FinkelS01} (input-languages of channels $c$ are included in $\Pref(u_cv_c^*)$ where $u_c,v_c$ are two words associated with $c$) and \emph{linear} FIFO nets \cite{choquet1987simulation} (input-languages are included in $\Pref(a_1^*a_2^*\ldots a_n^*)$ where each $a_i$ is a letter and $a_i \neq a_j$ iff $i \neq j$) both generalize Petri nets with still a decidable reachability problem.  
 A variant of the reachability problem, the deadlock problem, is shown decidable for input-\emph{letter}-bounded FIFO systems in \cite{gouda1987deadlock} by reducing to reachability for VASS, 
 but the extension to general input-bounded machines was left open.
 
 \emph{Flat}  machines are another subclass of input-bounded machines in which 
 the language of their control-graph, considered as a finite automaton, is a bounded language.
 For flat FIFO machines, control-state reachability is NP-complete \cite{DBLP:conf/lics/EsparzaGM12}; this result has recently been extended to reachability, channel unboundedness, and other classical properties \cite{DBLP:conf/concur/FinkelP19}.
 
 To the best of our knowledge, the decidability status of control-state reachability, reachability, deadlock, and termination was not known for input-bounded FIFO machines,
 which strictly include all the
 classes discussed above such as flat, input-letter-bounded, monogeneous, and linear FIFO machines (the last three types contain VASS and they are all incomparable).
 The unboundedness problem of input-bounded FIFO machines 
 was shown decidable in \cite{DBLP:journals/tcs/JeronJ93} by using the well-structured concepts but with no extension to decidability of reachability.
 
\paragraph{Contributions} 
Our contributions can be summarized as follows:
 \begin{itemize}
 	\item We solve a problem that was left open in \cite{gouda1987deadlock}, the decidability of the reachability problem for input-bounded FIFO machines. We present a simulation of input-bounded FIFO machines by counter machines with restricted zero tests. The main idea is to associate a counter with each word in the bounded language, and to ensure that the counters are incremented and decremented in a way that corresponds to the FIFO order. Since we can have repeated letters, and ambiguities in the FIFO machine, we first need to construct a normal form of the FIFO machine. Furthermore, we ensure that for every run in the FIFO machine, we can construct an equivalent run in the counter machine and vice-versa. 
 	\item
 	As we actually solve the general rational-reachability problem, we can deduce the decidability
 	of other verification properties like control-state reachability, deadlock, unboundedness, and termination.
 	
 	\item We unify various definitions from the literature, survey the (not well-known) results, and generalize them.
 	
 	\item We study the natural dual of the input-bounded reachability problem, which are systems of output-bounded languages in which the set of words received by each channel is constrained to be bounded, and are able to deduce the reachability, unboundedness, termination and control-state reachability for the same. 
 	
 	\item We obtain better upper bounds for the input-bounded reachability of FIFO machines with a single channel (reachability is still undecidable for FIFO machines with a single channel). This is done by reducing it to reachability in \emph{unary ordered multi-pushdown systems} (a class that was previously analyzed in \cite{DBLP:journals/ijfcs/AtigBH17}). It is, hence, solvable in EXPTIME.
 	
 	\item Following the bounded verification paradigm, applied to FIFO machines (for instance in \cite{DBLP:conf/lics/EsparzaGM12,DBLP:conf/concur/FinkelP19}), we open the way to a methodology that would apply existing results on input-bounded FIFO machines to general FIFO machines.
 \end{itemize}
 
\paragraph{Plan} In Section~\ref{sec:prel}, we present counter and FIFO machines, with the connection-deconnection protocol as an example. Section~\ref{sec:bounded-reachability} contains the main result, which states the decidability of rational-reachability for FIFO machines restricted to input-bounded languages. Sections~\ref{sec:variants} and \ref{sec:unboundterm}
 consider variants of the rational reachability problem such as deadlock, unboundedness and termination. In Section~\ref{sec:opbounded}, we consider the output-bounded reachability problem, and other variants such as boundedness and termination. We look at the input-bounded problems FIFO machines with a single channel, and obtain some lower bounds in Section~\ref{sec:singlechannel}.
 Finally, in Section~\ref{sec:conclusion}, we mention further results, state some open problems,
 and discuss a possible theory of boundable FIFO machines.
 
A preliminary version of this work has been presented at the 31st International Conference on Concurrency Theory (CONCUR '20).  Our contributions extend the conference version as follows:

 This work  additionally considers the symmetric case of output-bounded languages in which the set of words received by each channel is constrained to be bounded. We are able to deduce the reachability, unboundedness, termination and control-state reachability for the same. 
It also studies the special case of input-bounded FIFO machines with a single channel (reachability is still undecidable for general FIFO machines with a single channel). This is done by reducing it to reachability in unary ordered multi-pushdown systems (a class that was previously analyzed). It is, hence, solvable in EXPTIME.
 We also expand on the proof idea for boundedness and termination for input-bounded FIFO machines, and provide the explicit construction for the decidability of the same.

\section{Preliminaries}\label{sec:prel}

\subsection*{Words and Languages.}

Let $A$ be a finite alphabet. 
As usual, $A^\ast$ is the set of finite words over $A$,
and $A^+$ the set of non-empty finite words.
We let $|w|$ denote the length of $w \in A^\ast$.
For the empty word $\epsilon$, we have $|\epsilon| = 0$.
Given $a \in A$,
let $|w|_a$ denote the number of occurrences of $a$ in $w$.
With this, we let $\Alphw(w) = \{a \in A \mid |w|_a \ge 1\}$.
The concatenation of two words $u,v \in A^\ast$ is denoted by $u \cdot v$
or $u.v$ or simply $uv$.
The sets of prefixes, suffixes, and infixes of $w \in A^\ast$ are
denoted by $\Pref(w)$, $\Suf(w)$, and $\Inf(w)$, resp.
Note that $\{\varepsilon,w\} \subseteq \Pref(w) \cap \Suf(w) \cap \Inf(w)$.
For a set $X$, 
any mapping $f: A^\ast \to 2^X$ can be extended to
$f: 2^{A^\ast} \to 2^X$ letting, for
$L \subseteq A^\ast$, $f(L) = \bigcup_{w \in L} f(w)$.
In particular,
$\Alphw$, $\Pref$, $\Suf$, and $\Inf$ are extended in that way.

\begin{defiC}[\cite{10.2307/1994067}]
	Let $w_1,\ldots,w_n \in A^+$ be non-empty words where $n \ge 1$.
	A \emph{bounded language} over $(w_1,\ldots,w_n)$
	is a language
	$L \subseteq w_1^\ast \ldots w_n^\ast$.
\end{defiC}

We always assume that a bounded language $L$ is given together with
its tuple $(w_1,\ldots,w_n)$ and that $\Alphw(L) = \Alphw(w_1 \ldots w_n)$.
We say that $L$ is \emph{distinct-letter} if
$|w_1 \ldots w_n|_a \le 1$ for all $a \in A$. 
If $|w_1| = \cdots = |w_n| = 1$, i.e. $w_1, \ldots, w_n \in A$, then $L$ is a \emph{letter-bounded language.} Let us remark that the set of bounded languages is closed under $\Pref$ and $\Suf$.

\subsection*{Semi-Linear Sets.}
A \emph{linear} set $X$ (of dimension $d \geq 1$) is defined as a subset of $\mathbb{N}^d$ for which there exist a basis $\mathbf{b} \in \mathbb{N}^d$ and a finite set of periods $\{ \mathbf{p}_1, \ldots, \mathbf{p}_m\} \subseteq \mathbb{N}^d$ such that 
$X = \{ \mathbf{b} + \sum_{i=1}^{m} \lambda_i\mathbf{p}_i \mid \lambda_1, \ldots, \lambda_m \in \mathbb{N}\}$.
A \emph{semi-linear} set is defined as a finite union of linear sets.

\subsection*{Transition Systems.}

A \emph{labeled transition system} is a quadruple $\ts = (\tsstates,A,\tstrans,\tsinit)$ where $S$ is the (potentially infinite) set of
\emph{configurations}\footnote{We say \emph{configurations}
	rather than \emph{states}
	to distinguish them from the \emph{control states} used in FIFO and counter machines.}, $A$ is a finite alphabet, $\tsinit \in S$ is the \emph{initial configuration}, and ${\tstrans} \subseteq \tsstates \times A \times \tsstates$
is the \emph{transition relation}.

For $s,s' \in S$, let $s \tstrans s'$ if $s \tstransp{a} s'$ for some $a \in A$.
For $w \in A^\ast$, we write $s \tstransp{w} s'$ if there is a $w$-labeled
path from $s$ to $s'$. Formally, $s \tstransp{\epsilon} s'$ if $s = s'$, and
$s \tstransp{a w} s'$ if there is $\tsstateb \in S$ such that
$s \tstransp{a} \tsstateb$ and $\tsstateb \tstransp{w} s'$.
We let $\Traces{\ts} = \{w \in A^\ast \mid \tsinit \tstransp{w} s$ for some $s \in \tsstates\}$.

Given $w \in A^\ast$, we let $\reachsetL{\ts}{w} = \{s \in \tsstates \mid \tsinit \tstransp{w} s\}$.
Moreover, for $L \subseteq A^\ast$, $\reachsetL{\ts}{L} = \bigcup_{w \in L} \reachsetL{\ts}{w}$
is the set of configurations that are reachable via a word from $L$.
Finally, the \emph{reachability set} of $\ts$ is defined as $\reachset{\ts} = \reachsetL{\ts}{A^\ast}$. We call $\ts$ \emph{finite} if $\reachset{\ts}$ is finite (and this is the case if $\tsstates$ is finite). Otherwise, $\ts$ is called \emph{infinite}.

\subsection*{FIFO Machines.}

We consider FIFO machines having a sequential control graph rather than systems of communicating processes that are distributed systems. It is clear that, given a distributed system, one may compute the Cartesian product of all processes to obtain a FIFO machine (the converse is not always true).

\begin{defi}\label{def:fifomachine}
	A \emph{FIFO machine} is a tuple
	$\fifo = (\fifostates, \Ch{\nch}, \alphabet, \fifotransition, \init)$ where
	$\fifostates$ is a finite set of \emph{control states}, $\init \in \fifostates$ is an \emph{initial control state}, and
	$\Ch{\nch}$ is a finite set of \emph{channels}.
	Moreover, $\alphabet$ is a finite \emph{message alphabet}. It is partitioned into
	$\alphabet = \biguplus_{\ch \in \Ch{\nch}} \alphabet_\ch$ where $\alphabet_\ch$ contains the messages that can be sent
	through channel $\ch$. Finally, $\fifotransition \subseteq \fifostates \times \fifoAlpha{\fifo} \times \fifostates$ is
	a \emph{transition relation} where $\fifoAlpha{\fifo} = \{\sendaction{a}{\ch} \mid \ch \in \Ch{\nch}$ and $a \in \alphabet_\ch\} \cup \{\recaction{a}{\ch} \mid \ch \in \Ch{\nch}$ and $a \in \alphabet_\ch\}$ is the set of send and receive actions.
\end{defi}

\newcommand{\msgc}{e}

\begin{exa}[Connection-Deconnection Protocol]\label{conn-deconn} A model for the (simplified) connection-deconnection protocol, CDP, between two processes is described as follows (see Figure~\ref{fig:input-fifo-reach-word}): We model the protocol with two automata (representing the two processes) and two (infinite) channels. The first processes (on the left) can open a session (this is denoted by sending the message ``$a$'' through channel $\ch_1$ to the other process). Once a session is open, the first process can close it (by sending message ``$b$'' to the other process), or on the demand of the second process (if it receives the message ``$\msgc$''). This protocol has been studied in \cite{DBLP:conf/stacs/Jeron91}.
	
	In the example, it is natural to have two separate processes.
	However, following Definition~\ref{def:fifomachine}, we formalize this
	in terms of the Cartesian product of the two processes.
	That is, the CDP is modeled as the FIFO machine
	$\fifo = (\fifostates, \Ch{\nch}, \alphabet, \fifotransition, \init)$
	where $\fifostates = \{0,1\} \times \{0,1\}$ (the Cartesian product of the local state spaces)
	with initial state $\init = (0,0)$,
	$\Ch{\nch} = \{\ch_1,\ch_2\}$, $\Sigma = \Sigma_{\ch_1} \uplus \Sigma_{\ch_2}$ with $\Sigma_{\ch_1} = \{a,b\}$
	and $\Sigma_{\ch_2} = \{e\}$. Moreover, the transition relation $\fifotransition$ contains,
	amongst others, $((0,0),\sendaction{a}{c_1},(1,0))$ and $((1,0),\recaction{a}{c_1},(1,1))$.
	\exend
\end{exa}

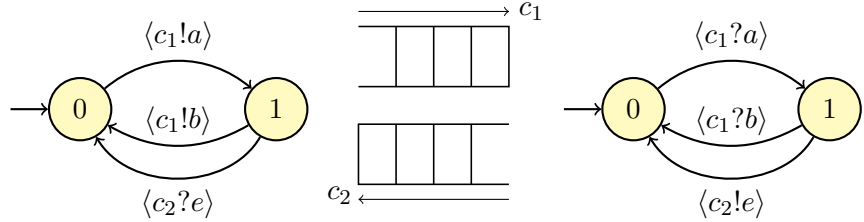
\begin{figure}[h]
	\centering
	\begin{tikzpicture}[->, node distance=1.75cm, auto, thick]
		\node[initial left, initial text=, place, fill=yellow!30] (p1) {$0$};
		\node[place, fill=yellow!30] (p2) [right= of p1]{$1$};
		
		\begin{scope}[shift={(0.7,1.1)}]
			\node[above] at (5.3,-0.05) {$c_1$};  
			\node (a) at (2.8, 0) {};
			\node (b) at (3.2, -0.4) {};
			\node (c) at (4.8, -0.42) {}; %{$b$};
			\node (d) at (4.3, -0.4) {};
			\draw[-,line width=0.6pt] %thickness
			(3.0,0) -- ++(2cm,0) --  ++(0,-0.8cm) -- ++(-2cm,0);
			\draw[-,line width=0.6pt]
			(4.0,0) -- ++ (0, -0.8cm);
			\draw[-,line width=0.6pt]
			(4.5,0) -- ++ (0, -0.8cm);
			\draw[-,line width=0.6pt]
			(3.5,0) -- ++ (0, -0.8cm);
			\draw[->, line width=0.4pt] 
			(3.0,0.2) -- ++(2cm,0);
		\end{scope}	
		
		\begin{scope}[shift={(0.7,1.8)}]	
			\node[above] at (2.75,-3.2) {$c_2$};  
			\node (a) at (3.3, -2.5) {}; %{$\msgc$};
			\node (b) at (3.2, -2.4) {};
			\node (c) at (4.8, -2.4) {};
			\node (d) at (5.2, -2.1) {};
			\draw[-,line width=0.6pt] %thickness
			(5.0,-2) -- ++(-2cm,0) --  ++(0,-0.8cm) -- ++(2cm,0);
			\draw[-,line width=0.6pt]
			(4.0,-2) -- ++ (0, -0.8cm);
			\draw[-,line width=0.6pt]
			(4.5,-2) -- ++ (0, -0.8cm);
			\draw[-,line width=0.6pt]
			(3.5,-2) -- ++ (0, -0.8cm);
			\draw[->, line width=0.4pt] 
			(5.0,-3) -- ++(-2cm,0); 
		\end{scope}

		\node[initial left, initial text=, place, fill=yellow!30] (q1)[right= 6.5cm of p1] {$0$};
		\node[place, fill=yellow!30] (q2)[right= of q1]{$1$};
		
		\path[->]
		(p1) edge [bend left=40] node[] {$\sendaction{a}{c_1}$} (p2)
		(p2) edge [bend left] node[swap] {$\sendaction{b}{c_1}$} (p1)
		(p2) edge [bend left=60] node[] {$\recaction{\msgc}{c_2}$} (p1)
		(q1) edge [bend left=40] node[] {$\recaction{a}{c_1}$} (q2)
		(q2) edge [bend left] node[swap] {$\recaction{b}{c_1}$} (q1)
		(q2) edge [bend left=60] node[] {$\sendaction{\msgc}{c_2}$} (q1)
		;
		
	\end{tikzpicture}
	\caption{The model of the connection-deconnection protocol} \label{fig:input-fifo-reach-word}
\end{figure}

A FIFO machine $\fifo = (\fifostates, \Ch{\nch}, \alphabet, \fifotransition, \init)$
induces a (potentially infinite) transition system
$\TSof{\fifo} = (\tscomp{\tsstates}{\fifo},\fifoAlpha{\fifo},\tscomp{\tstrans}{\fifo},\tscomp{\tsinit}{\fifo})$.
Its set of configurations is $\tscomp{\tsstates}{\fifo} = \fifostates \times \ChContents$.
In $(q,\contents) \in \tscomp{\tsstates}{\fifo}$, the
first component $q$
denotes the current control state and $\contents = (\contentsp{\ch})_{\ch \in \Ch{\nch}}$
determines the contents $\contentsp{\ch} \in \Sigma_\ch^\ast$ for every channel $\ch \in  \Ch{\nch}$.
The initial configuration is $\tscomp{\tsinit}{\fifo} = (\init,\initcontents)$
where $\initcontents = (\varepsilon,\ldots,\varepsilon)$, i.e., every
channel is empty.
The transitions are given as follows:
\begin{itemize}
	\item $(q,\contents) \Sendtransp{a}{\ch}{\fifo} (q',\contents')$ if
	$\sendtrans{a}{\ch}{q}{q'} \in \fifotransition$,
	$\contentspp{\ch} = \contentsp{\ch} \cdot a$,
	and $\contentspp{\chb} = \contentsp{\chb}$ for all $\chb \in \Ch{\nch} \setminus \{\ch\}$;
	\item $(q,\contents) \Rectransp{a}{\ch}{\fifo} (q',\contents')$ if
	$\rectrans{a}{\ch}{q}{q'} \in \fifotransition$,
	$\contentsp{\ch} = a \cdot \contentspp{\ch}$,
	and $\contentspp{\chb} = \contentsp{\chb}$ for all $\chb \in \Ch{\nch} \setminus \{\ch\}$.
\end{itemize}
The index $\fifo$ may be omitted whenever $\fifo$ is clear from the context.

The \emph{reachability set} of $\fifo$ is defined as
the reachability set of $\TSof{\fifo}$, i.e.,
$\reachset{\fifo} = \reachset{\TSof{\fifo}}$ and,
for $L \subseteq \fifoAlpha{\fifo}^\ast$,
$\reachsetL{\fifo}{L} = \reachsetL{\TSof{\fifo}}{L}$.
Moreover, we let $\Traces{\fifo} = \Traces{\TSof{\fifo}}$.

\begin{exa}
	An example run of the FIFO machine $\fifo$ from Example~\ref{conn-deconn}
	and Figure~\ref{fig:input-fifo-reach-word} is
	$((0,0), (\varepsilon, \varepsilon)) \xrightarrow{\sendaction{a}{\ch_1}} ((1,0), (a, \varepsilon)) \xrightarrow{\recaction{a}{\ch_1}} ((1,1), (\varepsilon, \varepsilon)) \xrightarrow{\sendaction{\msgc}{\ch_2}} ((1,0), (\varepsilon, \msgc))$.
	As for the reachability set, we have, e.g.,
	$((1,1),((ba)^\ast, \epsilon))) \subseteq \reachset{{\fifo}}$ and
	$((0,0),(b(ab)^\ast, \msgc)) \subseteq \reachset{{\fifo}}$.
	Let us remark that CDP is not half-duplex because there are reachable configurations with both channels non-empty, e.g., $((0,0), (b,\msgc))$; moreover, it is neither monogeneous, nor linear, nor input-letter-bounded.
	\exend
\end{exa}

\newcommand{\proj}[2]{\mathit{proj}_{#2}(#1)}
\newcommand{\projsend}[2]{\mathit{proj}_{#2!}(#1)}
\newcommand{\projrec}[2]{\mathit{proj}_{#2?}(#1)}

\newcommand{\projsendmap}[1]{\mathit{proj}_{#1!}}
\newcommand{\projrecmap}[1]{\mathit{proj}_{#1?}}

\newcommand{\tupleL}{\bar{L}}

\newcommand{\sendmap}[1]{\mathit{sent}_{#1}}
\newcommand{\recmap}[1]{\mathit{recd}_{#1}}

\subsection*{Counter Machines.} We next recall the notion of counter machines,
where multiple counters can take non-negative integer values, be incremented
and decremented, and be tested for zero (though in a restricted fashion).

\newcommand{\counter}{x}
\newcommand{\counterb}{y}

\begin{defi}
	A \emph{counter machine} (with zero tests) is a tuple
	$\CS = (\fifostates, \Cnt, \fifotransition, \init)$.
	Like in a FIFO machine, $\fifostates$ is the finite set of \emph{control states}
	and $\init \in \fifostates$ is the \emph{initial control state}.
	Moreover, $\Cnt$ is a finite set of \emph{counters}
	and $\fifotransition \subseteq \fifostates \times \csAlpha{\CS} \times \fifostates$
	is the transition relation where $\csAlpha{\CS} =
	\{\inc{\counter},\dec{\counter} \mid \counter \in \Cnt\} \times 2^\Cnt$.
\end{defi}

The counter machine $\CS$ induces a transition system
$\TSof{\CS} = (\tscomp{\tsstates}{\CS},\csAlpha{\CS},\tscomp{\tstrans}{\CS},\tscomp{\tsinit}{\CS})$
with set of configurations $\tscomp{\tsstates}{\CS} = \fifostates \times \N^\Cnt$.
In $(q,\chcontents) \in \tscomp{\tsstates}{\CS}$,
$q$ is the current control state and
$\cscontents = (\cscontentsp{\counter})_{\counter \in \Cnt}$
represents the counter values.
The initial configuration is
$\tscomp{\tsinit}{\CS} = (\init,\initcntcontents)$
where $\initcntcontents$ maps all counters to $0$.
For $\incdec \in \{\mathsf{inc},\mathsf{dec}\}$, $\counter \in \Cnt$, and $Z \subseteq \Cnt$ (the counters tested for zero), there is a transition
$
(q,\cscontents) \incdectransp{\counter}{Z}{\fifo}_\CS (q',\cscontents')
$
if $\incdectrans{\counter}{Z}{q}{q'} \in \fifotransition$,
$\cscontentsp{\counterb} = 0$ for all $\counterb \in Z$ (applies the zero tests),
$\cscontentspp{\counter} = \cscontentsp{\counter} + 1$ if $\incdec = \mathsf{inc}$ and
$\cscontentspp{\counter} = \cscontentsp{\counter} - 1$ if $\incdec = \mathsf{dec}$, and
$\cscontentspp{\counterb} = \cscontentsp{\counterb}$ for all $\counterb \in \Cnt \setminus \{\counter\}$.

\medskip

\newcommand{\BCounters}[1]{L_{#1}^{\textsf{\textup{zero}}}}

The \emph{reachability set} of $\CS$ is defined as
$\reachset{\CS} = \reachset{\TSof{\CS}}$.
For $L \subseteq \csAlpha{\CS}^\ast$, we also let
$\reachsetL{\CS}{L} = \reachsetL{\TSof{\CS}}{L}$.
Moreover, $\Traces{\CS} = \Traces{\TSof{\CS}}$.
To get decidability of reachability in counter machines, we impose the restriction
that, once a counter has been tested for zero, it cannot be incremented or decremented anymore. This is clearly an extension of VASS\@.
To define this, let $\smash{\BCounters{\CS}} $ be the set of words
$(\incdec_1(\counter_1),Z_1) \ldots (\incdec_n(\counter_n),Z_n) \in \csAlpha{\CS}^\ast$
such that, for every two positions $1 \le i \le j \le n$, we have $\counter_j \not\in Z_i$.

\begin{thm}\label{thm:reach-counters}
	The following problem is decidable (though inherently non-elementary):
	Given a counter machine $\CS = (\fifostates, \Cnt, \fifotransition, \init)$,
	a regular language $L \subseteq \csAlpha{\CS}^\ast$, a control state
	$q \in \fifostates$, and a semi-linear set $V \subseteq \N^\Cnt$,
	do we have $(q,\cscontents) \in \reachsetL{\CS}{\BCounters{\CS} \cap L}$
	for some $\cscontents \in V$?
\end{thm}

\begin{proof}[Proof sketch]
	Reachability in presence of a semi-linear target set and
	restricted zero tests
	straightforwardly reduces to configuration-reachability
	in counter machines without zero tests (i.e., VASS and Petri nets).
	The latter is decidable \cite{Mayr84}, though inherently
	non-elementary \cite{CzerwinskiLLLM19}.
	First, zero tests are postponed to the very end of an execution
	and, to this aim, stored in the control-state.
	Second, to check whether a counter valuation is contained in $V$,
	we can branch, whenever we are in the given control-state $q$,
	into a new component that decrements counters accordingly and
	eventually checks whether they are all zero.
\end{proof}

\section{The Input-Bounded Rational-Reachability Problem}\label{sec:bounded-reachability}

It is very well known that the following reachability problem is undecidable:
Given a FIFO machine
$\fifo = (\fifostates, \Ch{\nch}, \alphabet, \fifotransition, \init)$,
a configuration $(q,\contents) \in \tscomp{\tsstates}{\fifo}$,
and a regular language $L \subseteq \fifoAlpha{\fifo}^\ast$, do we have
$(q,\contents) \in \reachsetL{\fifo}{L}$?
Of course, the problem is already undecidable when we impose $L = \fifoAlpha{\fifo}^\ast$.
Motivated by this negative result, we are looking for language classes
$\mathfrak{C}$ that render the problem decidable under the
restriction that $L \in \mathfrak{C}$.

We say that a FIFO machine $\fifo = (\fifostates, \Ch{\nch}, \alphabet, \fifotransition, \init)$ has a \emph{bounded reachability set} if there is a tuple $(L_\ch)_{\ch \in \Ch{\nhc}}$ of regular bounded languages $L_\ch \subseteq \Sigma_\ch^\ast$ such that, for all $(q, \contents) \in \reachset{\fifo}$, we have $\contents \in \prod_{\ch \in \Ch{\nch}} L_\ch$. We observe that restricting the reachability set to be bounded is not sufficient to obtain a decidable reachability problem. We show this by simulating any two counter Minsky machine by a FIFO machine with fixed languages $L_\ch$.

\begin{thm}
	\label{thm:reachbound}
	The reachability problem is undecidable for FIFO machines with a (given) bounded reachability set.
\end{thm}

\begin{proof} \label{app:reachbound-proof}
	
	We prove this by 
	simulating a (two) counter Minsky machine by a FIFO machine with a bounded reachability set. Intuitively, if the two counters have values $x_1$ and $x_2$ at a state $q$, the FIFO machine is at state $q$ with channel contents $\$a^{x_1}\#b^{x_2}\&$.  
	
	Consider a Minsky machine $\CS = (Q, \Cnt,T, \init)$, where $Q$ is the set of states, $\Cnt = \{x_1, x_2\}$ is the set of counters, $\init$ the initial state, and  $T = \{\delta_1, \ldots, \delta_n\}$ is the set of transition rules, which can be of two types from any state $\hat{q} \in Q$:

	\begin{itemize}
		\item $\delta_i: x_j:= x_j +1;  \texttt{ goto } q;$
		\item $\delta_i: \texttt{if } x_j>0  \texttt{ then } (x_j:=x_j-1; \texttt{ goto } q)  \texttt{ else}  \texttt{ goto } q';$
	\end{itemize}
	%$\delta_i: x_j:= x_j +1;  \texttt{ goto } q_k;$ 
	for $j = \{1,2\}$ and $q,q' \in Q$.	
	
	\smallskip
	
	We can construct a FIFO machine $\fifo = (Q', \Ch{\nch}, \Sigma, \fifotransition, q_0)$ with a bounded reachability set as follows: $Q' = Q \uplus Q''$, where $Q''$ is a set of intermediate states; $\Sigma = \{a, b, \#, \$, \&\}$. There is a single channel, hence $|\Ch{\nch}|=1$. We consider the language $L = \$^*a^*\#^*b^*\&^*\$^*a^*\#^*b^*\&^*$, and we start with queue contents as $\$\#\&$, where the letters $\#, \$, \&$ are used as markers during the test for zero. The number of occurrences of the letter $a$ (resp.\ $b$) corresponds to the counter valuations for $x_1$ (resp.\ $x_2$) in the configurations of the Minsky machine, and we rotate the tape contents in such a way that the contents always belong to $L$. For every state $q'' \in Q$ and every transition from it with the rule $\delta_i: x_1:= x_1 +1;  \texttt{ goto } q;$ for some $q \in Q$, we create the following transition sequence. We ensure that we read the markers, and add an additional letter $a$ to the channel. We see that at every intermediate configuration, the channel contents still belongs to $L$.
	
	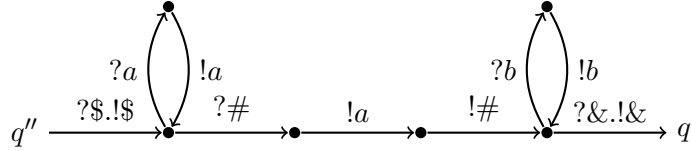
\begin{figure}[h]
		\centering
		\begin{tikzpicture}[->, node distance=1.5cm, auto, thick]
			\node[] (p1) {$q''$};
			%\node[] (p2) [right= of p1]{$q'_{i,1}$};
			\node[] (p3) [right= of p1, circle, fill, inner sep = 1.5pt] {};
			\node[] (p4) [above= of p3, circle, fill, inner sep = 1.5pt] {};
			\node[] (p5) [right= of p3, circle, fill, inner sep = 1.5pt] {};
			\node[] (p6) [right= of p5, circle, fill, inner sep = 1.5pt] {};
			\node[] (p7) [right= of p6, circle, fill, inner sep = 1.5pt] {};
			\node[] (p8) [above= of p7, circle, fill, inner sep = 1.5pt] {};
			\node[] (p9) [right= of p7] {$q$};
			\path[->]
			(p1) edge node[] {$?\$.!\$$} (p3)
	
			(p3) edge [bend left] node[] {$?a$} (p4)
			(p4) edge [bend left] node[] {$!a$} (p3)
			(p3) edge node[] {$?\#$} (p5)
			(p5) edge node[] {$!a$} (p6)
			(p6) edge node[] {$!\#$} (p7)
			(p7) edge [bend left] node[] {$?b$} (p8)
			(p8) edge [bend left] node[] {$!b$} (p7)
			(p7) edge node[] {$?\&.!\&$} (p9)
			;
			
		\end{tikzpicture}
		\caption{Incrementing $x_1$} \label{eq-minsky}
	\end{figure}
	
	Likewise, a transition of the form
	\begin{center}
		$\delta_i: \texttt{if } x_1>0  \texttt{ then } (x_1:=x_1-1; \texttt{ goto } q)$  $ \texttt{ else}  \texttt{ goto } q';$
	\end{center}
	can be constructed as show in Figure~\ref{eq-minsky-2}. Here, in the first case ($x_1 > 0$) we ensure that at least one letter $a$ is read from the channel contents, and rewrite the rest of the contents as is. If $x_1=0$, then we read only the end markers (and letters $b$, if any) and rewrite them once again.
	
	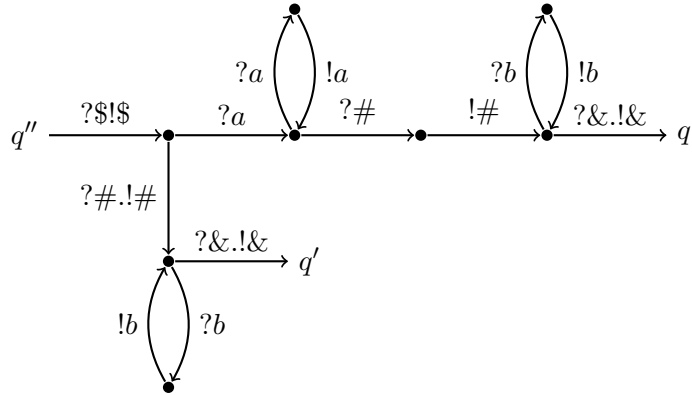
\begin{figure}[!h]
		\centering
		\begin{tikzpicture}[->, node distance=1.5cm, auto, thick]
			\node[] (p1) {$q''$};
		
			\node[] (p6) [right= of p1, circle, fill, inner sep = 1.5pt] {};
			\node[] (p3) [right= of p6, circle, fill, inner sep = 1.5pt] {};
			\node[] (p4) [above= of p3, circle, fill, inner sep = 1.5pt] {};
			\node[] (p5) [right= of p3, circle, fill, inner sep = 1.5pt] {};
		
			\node[] (p7) [right= of p5, circle, fill, inner sep = 1.5pt] {};
			\node[] (p8) [above= of p7, circle, fill, inner sep = 1.5pt] {};
			\node[] (p9) [right= of p7] {$q$};
			\node[] (p21) [below= of p6, circle, fill, inner sep = 1.5pt]{};
			\node[] (p23) [below= of p21, circle, fill, inner sep = 1.5pt]{};
			\node[] (p22) [right= of p21] {$q'$};
			\path[->]
	
			(p1) edge node[] {$?\$!\$$} (p6)
		
			(p6) edge node[] {$?a$} (p3)
		
			(p3) edge [bend left] node[] {$?a$} (p4)
			(p4) edge [bend left] node[] {$!a$} (p3)
			(p3) edge node[] {$?\#$} (p5)
		
			(p5) edge node[] {$!\#$} (p7)
			(p7) edge [bend left] node[] {$?b$} (p8)
			(p8) edge [bend left] node[] {$!b$} (p7)
			(p7) edge node[] {$?\&.!\&$} (p9)
			(p6) edge node[swap] {$?\#.!\#$} (p21)
			(p21) edge [bend left] node[] {$?b$} (p23)
			(p23) edge [bend left] node[] {$!b$} (p21)
			(p21) edge node[] {$?\&.!\&$} (p22)
			;
			
		\end{tikzpicture}
		\caption{Decrementing $x_1$} \label{eq-minsky-2}
	\end{figure}
	
	Similar constructions can be made for all transitions, and the queue contents are always contained in $\$^*a^*\#^*b^*\&^*\$^*a^*\#^*b^*\&^*$. Hence, the reachability set of the FIFO machine is letter-bounded (but it is not distinct-letter).
	\end{proof}

\begin{rem}
	We see, in the previous proof, that the input language of the above machine is not bounded: If we have a transition from $q$ in the original machine of the kind $\delta_i: x_j:= x_j +1;  \texttt{ goto } q;$ (a loop), the input language of the machine would be $(\$^*a^*\#b^*\&)^*$ for this transition, which is not bounded. Furthermore, the machine is not flat either, since there can be control states that are in more than one elementary loop.
\end{rem}

\newcommand{\actionword}{\alpha}

We therefore consider a different restriction to obtain decidability.
For a given FIFO machine $\fifo = (\fifostates, \Ch{\nch}, \alphabet, \fifotransition, \init)$, we are interested in $\reachsetL{\fifo}{L}$ where
$L \subseteq \fifoAlpha{\fifo}^\ast$ is \emph{input-bounded} in the following sense:
For every channel $\ch$, the sequence of messages that are sent
through channel $\ch$ is from a given regular bounded language $L_\ch \subseteq \Sigma_\ch^\ast$.

Let us be more formal.
For $\ch \in \Ch{\nch}$, we let $\projsendmap{\ch}: \fifoAlpha{\fifo}^\ast \to \Sigma_\ch^\ast$ be
the homomorphism
defined by
$\projsend{\sendaction{a}{\ch}}{\ch} = a$ for all $a \in \Sigma_\ch$, and
$\projsend{\action}{\ch} = \varepsilon$ if $\action \in \fifoAlpha{\fifo}$ is not of the form
$\sendaction{a}{\ch}$ for some $a \in \Sigma_\ch$. 
We define $\projrecmap{\ch}: \fifoAlpha{\fifo}^\ast \to \Sigma_\ch^\ast$ accordingly.

With this,
given a tuple $\boundedL = (L_\ch)_{\ch \in \Ch{\nhc}}$ of bounded languages
$L_\ch \subseteq \Sigma_\ch^\ast$,
we set
$\sendL = \{\runFIFO \in \fifoAlpha{\fifo}^\ast \mid
\projsend{\runFIFO}{\ch} \in L_\ch$
for all $\ch \in \Ch{\nch}\}$ 
and
$\boundedL_? = \{\runFIFO \in \fifoAlpha{\fifo}^\ast \mid
\projrec{\runFIFO}{\ch} \in L_\ch$
for all $\ch \in \Ch{\nch}\}$.
We observe that, if all $L_\ch$ are regular, then so are $\sendL$ and $\recL$.

\begin{defi}\label{def:singleconf}
	The \emph{input-bounded (IB) reachability problem} asks whether a given configuration
	$(q,\contents)$ is reachable
	along a sequence of actions from $\sendL$, i.e.,
	whether $(q,\contents) \in \reachsetL{\fifo}{\sendL}$.
\end{defi}

Note that, if $(\init,\initcontents) \xrightarrow{ \runFIFO }_\fifo (q, \contents)$
and $\runFIFO \in \sendL$, then we also have $\runFIFO \in \Pref(\recL)$
due to the FIFO policy. Thus,
$\reachsetL{\fifo}{\sendL} = \reachsetL{\fifo}{\sendL \cap \Pref(\recL)}$
so that we can restrict to action sequences from $\sendL \cap \Pref(\recL)$.
We will call $\sendL \cap \Pref(\recL)$ the set of \emph{valid} words.

\begin{exa}\label{ex:boundable}
	Let us come back to the protocol CDP $\fifo$ from Example~\ref{conn-deconn} and Figure~\ref{fig:input-fifo-reach-word},
	which is neither monogeneous nor linear nor flat.
	Since the ``input-languages'' of the two channels
	(i.e.\ the languages of words that
	record the messages entering a channel)
	contain $\{a,ab\}^*$ 
	and $\msgc^\ast$, resp., and since $\{a,ab\}^*$ is
	not a bounded language,
	we have $\Traces{\fifo} \not\subseteq \sendL$ for
	every pair of bounded languages $\boundedL$.
	In other words, $\fifo$ is not input-bounded.
	However, when we look at the reachability set obtained by considering the tuple of bounded languages $\boundedL = (L_{\ch_1}, L_{\ch_2})$ where $L_{\ch_1} = (ab)^*(a+\varepsilon)(ab)^*$ is a bounded language over $(ab,a,ab)$,
	and $L_{\ch_2} = \msgc^\ast$ is a bounded language over $(e)$, we still obtain the entire reachability set.
	That is, we have $\reachset{\fifo} = \reachsetL{\fifo}{\sendL}$.
	Hence, even though the input-languages of the system are not all bounded,
	we can still compute the reachability set by restricting our exploration
	to a tuple of (regular) bounded languages $\boundedL$.
	\exend
\end{exa}

Actually, instead of reachability of a single configuration as stated
in Definition~\ref{def:singleconf}, we study
a more general problem, called the \emph{input-bounded rational-reachability problem}.
It asks whether
a configuration $(q,\contents)$ is reachable for some channel contents $\contents$
from a given \emph{rational} relation. So let us define rational relations.

\subsection*{Rational and Recognizable Relations.}

Consider a relation $\Rel \subseteq \ChContents$.
We say that $\Rel$ is \emph{rational}
if there is a regular word language $R \subseteq \Theta^\ast$ over the alphabet
$\Theta = \prod_{\ch \in \Ch{\nch}}  (\Sigma_\ch \cup \{\epsilon\})$ such that
$\Rel = \{{(\aletter_\ch^1 \cdot \cdots \cdot \aletter_\ch^n)}_{\ch \in \Ch{\nch}}
\mid \aletter^1 \cdots \aletter^n \in R$ with $n \in \N$ and $\aletter^i = {(\aletter^i_\ch)}_{\ch \in \Ch{\nch}} \in \Theta$ for $i \in \{1,\ldots,n\}\}$. Here,
$\aletter_\ch^1 \cdot \cdots \cdot \aletter_\ch^n \in \Sigma_\ch^\ast$ is the concatenation of
all $\aletter_\ch^i \in \Sigma_\ch \cup \{\epsilon\}$ while ignoring the neutral element $\epsilon$. For example, in the presence of two channels, $\Rel = \{(a^m,b^n) \mid m \ge n\}$ is a rational relation, witnessed by
$R = ((a,b) + (a,\epsilon))^\ast$.
In the following, we will always assume that a rational relation is given in terms
of a finite automaton for the underlying regular language $R$.

A relation $\Rel \subseteq \ChContents$ is called \emph{recognizable}
if it is the finite union of relations of the form
$\prod_{\ch \in \Ch{\nch}} R_\ch$ where all $R_\ch \subseteq \Sigma_\ch^\ast$
are regular languages. Note that every recognizable relation is rational while the converse is, in general, false.

We define the \emph{Parikh image} of a relation $\Rel \subseteq \ChContents$
as $\Parikhimg{\Rel} = \{(\pi_a)_{a \in \Sigma} \in \N^\Sigma \mid \exists\contents
= (\contentsp{\ch})_{\ch \in \Ch{\nch}} \in \Rel:
\pi_a = |\contents_\ch|_a$ for all $\ch \in \Ch{\nch}$ and $a \in \Sigma_\ch\}$.
It is well known that, if $\Rel$ is rational, then $\Parikhimg{\Rel}$ is semi-linear.

For more background on rational relations and their subclasses, we refer to \cite{Berstel79,Choffrut06}.

\subsection*{The \IBounded Rational-Reachability Problem.}

We are now prepared to define the input-bounded (IB) rational-reachability problem and to state its decidability:

\begin{defi}
	The \emph{\IBounded rational-reachability problem} is defined as follows:
	Given a FIFO machine $\fifo = (\fifostates, \Ch{\nch}, \alphabet, \fifotransition, \init)$,
	a tuple $\boundedL = (L_\ch)_{\ch \in \Ch{\nhc}}$ of non-empty regular bounded languages
	$L_\ch \subseteq \Sigma_\ch^\ast$ (each given in terms of a finite automaton),
	a control state $q \in Q$, and a rational relation $\Rel \subseteq \ChContents$.
	Do we have $(q,\contents) \in \reachsetL{\fifo}{\sendL}$ for some $\contents \in \Rel$?
\end{defi}

\begin{thm}\label{thm:general-I-bounded-reach}
	\IBounded rational-reachability is decidable for FIFO machines.
\end{thm}

The remainder of this section is devoted to the proof of Theorem~\ref{thm:general-I-bounded-reach}.

\smallskip

Let $\fifo = (\fifostates, \Ch{\nch}, \alphabet, \fifotransition, \init)$ and
let $\boundedL = (L_\ch)_{\ch \in \Ch{\nhc}}$ be a tuple of
non-empty regular bounded languages
$L_\ch \subseteq \Sigma_\ch^\ast$ over $(w_{\ch,1}, \ldots, w_{\ch,n_\ch})$.
We proceed by reduction to counter machines. The rough idea is to represent
the contents of channel $\ch$ in terms of several counters, one for every
component $w_{\ch,i}$. To have a faithful simulation, we rely on a normal form
of $\fifo$ and its bounded languages,
which can be achieved at the expense of an exponential blow-up
of the FIFO machine.

\newcommand{\cAlph}[1]{\Sigma_{#1}}
\newcommand{\ciAlph}[2]{\Sigma_{#1,#2}}

\newcommand{\Valid}{\textscrb{V}}

\begin{defi}\label{def:normal-form}
	We say that $\fifo$ and $\boundedL$ are in \emph{normal form} if the following hold:
	\begin{enumerate}
		\item For all $\ch \in \Ch{\nch}$, $\Sigma_\ch \subseteq \Alphw(L_c)$ and $L_\ch$ is distinct-letter.
		
		\item
		We have $\Traces{(\fifostates,\fifoAlpha{\fifo},\fifotransition,\init)} \subseteq \Pref(\Valid)$ where $\Valid = \sendL \cap \Pref(\recL)$.
		Note that $(\fifostates,\fifoAlpha{\fifo},\fifotransition,\init)$ is the finite transition system induced by
		the control graph of $\fifo$.
	\end{enumerate}
\end{defi}

Given a FIFO machine $\hat \fifo = (\hat\fifostates, \Ch{\nch}, \hat\alphabet, \hat\fifotransition, \hat q_0)$
and the tuple $\hat{\boundedL} = (\hat L_\ch)_{\ch \in \Ch{\nhc}}$ of non-empty regular bounded languages
$\hat L_\ch \subseteq \hat\Sigma_\ch^\ast$,
we now construct $\fifo = (\fifostates, \Ch{\nch}, \alphabet, \fifotransition, \init)$
and $\boundedL = {(L_\ch)}_{\ch \in \Ch{\nhc}}$
in normal form such that a reachability query in the former can be transformed into a reachability
query in the latter (made precise in Lemma~\ref{lem:normal-form} below).

\subsection*{Distinct-Letter Property.}

Consider the bounded language $\hat L_\ch$
over $(\hat w_{\ch,1}, \ldots, \hat w_{\ch,n_\ch})$.
For $i \in \{1,\ldots,n_\ch\}$, let $m_i = |\hat w_{\ch,1}| + \ldots + |\hat w_{\ch,i}|$
be the number of letters in the first $i$ words.
Moreover, $m = m_{n_\ch}$.
Let $\Sigma_\ch$ denote the alphabet $\{a_1^\ch,\ldots,a_{m}^\ch\}$.
It contains the ``distinct'' letters for the bounded language $L_\ch$
over $(w_{\ch,1}, \ldots, w_{\ch,n_\ch})$, where we let
$w_{\ch,1} = a_1^\ch \ldots a_{m_1}^\ch$ and
$w_{\ch,i} = a_{m_{i-1}+1}^\ch \ldots a_{m_i}^\ch$ for
$i \ge 2$.
In other words, the letters in $(w_{\ch,1}, \ldots, w_{\ch,n_\ch})$ are numbered consecutively.
In order to obtain the language $L_\ch$, we first consider the homomorphism
$h_\ch: \Sigma_\ch^\ast \to \hat\Sigma_\ch^\ast$ where
$h_\ch(a_i^\ch)$ is the $i$-th letter in the word $\hat w_{\ch,1} \ldots \hat w_{\ch,n_\ch}$.
We obtain $L_\ch$ as
$h_\ch^{-1}(\hat L_\ch) \cap (w_{\ch,1})^\ast \ldots (w_{\ch,n_\ch})^\ast$,
hence preserving regularity and boundedness.
We then remove those words from $(w_{\ch,1}, \ldots, w_{\ch,n_\ch})$
(and their letters from $\Sigma_\ch$) whose
letters do not occur in $L_\ch$.
We have $\Sigma_\ch \subseteq \Alphw(L_c)$.

\newcommand{\sname}[3]{$\begin{array}{c}#1\\[-0.8ex]{\scalebox{0.7}{#2,#3}}\end{array}$}
\newcommand{\Aname}{\textup{M}}
\newcommand{\Bname}{\textup{L}}
\newcommand{\Cname}{\textup{R}}
\newcommand{\decy}{$\begin{array}{c}\dec{y}\\[-0.5ex]x=0\end{array}$}

\begin{exa}\label{ex:nformone}
	For example, suppose we have one channel $\ch$ and $\hat L_\ch = (ab)^\ast bb^\ast$
	over $\bpair{ab}{b}$. We determine the language $L_\ch$
	over $\bpair{a_1 a_2}{a_3}$ (omitting the superscript $\ch$ in the letters).
	The homomorphism $h_\ch: \{a_1,a_2,a_3\}^\ast \to \{a,b\}^\ast$
	is given by $h_\ch(a_1) = a$ and $h_\ch(a_2) = h_\ch(a_3) = b$.
	We have $h_\ch^{-1}(\hat L_\ch) =
	(a_1(a_2 + a_3))^\ast(a_2 + a_3)(a_2 + a_3)^\ast$, which we intersect with
	$(a_1a_2)^\ast a_3^\ast$. We thus get the regular bounded language $L_\ch =
	(a_1a_2)^\ast a_3a_3^\ast$ over $\bpair{a_1a_2}{a_3}$.
	All letters from $\{a_1,a_2,a_3\}$ occur in $L_\ch$ so that we are done.
\end{exa}

\subsection*{Trace Property.}

\newcommand{\newa}{\textscrb{e}}

In the next step, we build the FIFO machine $\fifo = (\fifostates, \Ch{\nch}, \alphabet, \fifotransition, \init)$ such that
$\Traces{(\fifostates,\fifoAlpha{\fifo},\fifotransition,\init)} \subseteq \Pref(\Valid)$ with
$\Valid = \sendL \mathrel{\cap} \Pref(\recL)$.
First, to take care of the homomorphisms $h_\ch$, we
define the transition relation $h^{-1}({\hat\fifotransition})
= \{\sendtrans{\newa}{\ch}{q}{q'} \mid \sendtrans{a}{\ch}{q}{q'} \in \hat\fifotransition$ and $\newa \in h_\ch^{-1}(a)\}
\mathrel{\cup}
\{\rectrans{\newa}{\ch}{q}{q'} \mid \rectrans{a}{\ch}{q}{q'} \in \hat\fifotransition$ and $\newa \in h_\ch^{-1}(a)\}$.
Thus, the set of actions of $\fifo$ will be $\fifoAlpha{\fifo} = \{\sendaction{\newa}{\ch} \mid \ch \in \Ch{\nch}$ and $\newa \in \alphabet_\ch\} \cup \{\recaction{\newa}{\ch} \mid \ch \in \Ch{\nch}$ and $\newa \in \alphabet_\ch\}$.

To continue our above example,
a transition $\sendtrans{b}{\ch}{q}{q'}$
would be replaced with the two transitions
$\sendtrans{a_2}{\ch}{q}{q'}$ and $\sendtrans{a_3}{\ch}{q}{q'}$, and
similarly for $\rectrans{b}{\ch}{q}{q'}$.

To guarantee trace inclusion in $\Pref(\Valid)$,
we will consider a deterministic (not necessarily complete)
finite automaton $\A = (\fifostates_\A,\fifoAlpha{\fifo},\fifotransition_\A,q^0_\A,F_\A)$,
with set of final states $F_\A \subseteq Q_\A$, whose language is $L(\A) = \Valid$
and where, from every state, a final state is reachable
in the finite graph $(Q_\A,\fifotransition_\A)$.
With this, we define $\fifo$ as the product of the FIFO machine
$h^{-1}(\hat\fifo) = (\hat\fifostates, \Ch{\nch}, \alphabet, h^{-1}(\hat\fifotransition), \hat q_0)$
and $\A$ in the expected manner. In particular, the set of control states of $\fifo$ is $\hat Q \times Q_\A$, and its initial state
is the pair $(\hat q_0,q^0_\A)$.

\begin{exa}\label{ex:normal-form}
	Figure~\ref{fig:counter-from-fifo} illustrates the result of the normalization
	procedure for a FIFO machine $\hat M$ with one single channel $\ch$ (which is therefore
	omitted) and its bounded language
	$\hat L_\ch = (ab)^\ast bb^\ast$ over $\bpair{ab}{b}$.
	Recall from Example~\ref{ex:nformone} that the corresponding homomorphism $h_\ch$ maps $a_1$ to $a$ and both $a_2$
	and $a_3$ to $b$, and that we obtain $L_\ch = (a_1a_2)^\ast a_3a_3^\ast$. Moreover, $\fifo$ is
	the product of $h^{-1}(\hat\fifo)$ (depicted in the top center)
	and a finite automaton $\A$ for $\Valid = \sendL \cap \Pref(\recL)$
	(obtained as the shuffle of the two finite automata on the top right).
	The state names in $\fifo$ reflect the states of $\hat\fifo$ and $\A$ they originate
	from. We depict only accessible states of $\fifo$ from which we can still
	complete the word read so far to a word in $\Valid$.
	For example, $(q_1,\Aname,\Bname)$ and $(q_1,\Bname,\Cname)$
	would no longer allow us to reach the final state
	\Cname~of the $\sendL$-component.
	\exend
\end{exa}

\newcommand{\tinyplus}{{++}}
\newcommand{\tinyminus}{{--}}
\usetikzlibrary{positioning}

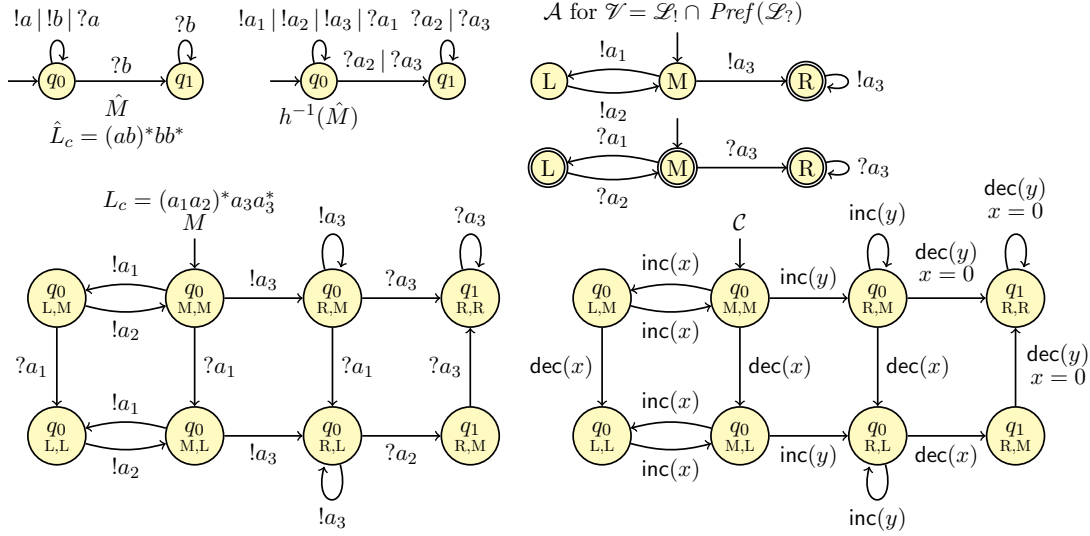
\begin{figure}[t]
	\hspace{-2em}
	\scalebox{0.80}{
		\begin{tikzpicture}[->, node distance=1.5cm, auto, thick, every state/.style={circle,inner sep=0pt, minimum size=0.6cm}]
			
			\node[initial left, initial text=, state, fill=yellow!30] (a0) {$q_0$};
			\node[state, fill=yellow!30] (a1) [right= of a0]{$q_1$};
			\node (b) at (1, -0.4) {$\hat\fifo$};
			\node (b) at (1, -0.9) {$\hat L_\ch = (ab)^\ast bb^\ast$};
			\node (b) at (2.2, -2) {$L_\ch = (a_1a_2)^\ast a_3a_3^\ast$};
			
			\path[->]
			(a0) edge [loop above] node[] {$!a\,|\,!b\,|\,?a$} (a0)
			(a0) edge node[] {$?b$} (a1)
			(a1) edge [loop above] node[] {$?b$} (a1)
			;
			
			\node[initial left, initial text=, state, fill=yellow!30] (q0) [right=1.6cm of a1] {$q_0$};
			\node[state, fill=yellow!30] (q1) [right= of q0]{$q_1$};
			\node[below=0.2cm, align=center] at (q0) { ~~~~~~~~~~~~~~~~~~~~$h^{-1}(\hat\fifo)$ };
			
			\node[initial above, initial text=, state, fill=yellow!30] [right=3.2cm of q1] (r0) {\Aname};
			\node[state, fill=yellow!30] (r0p) [left= of r0] {\Bname};
			\node[state, fill=yellow!30, accepting] (r1) [right= of r0] {\Cname};
			\node[above=0.8cm, align=center] at (r0) { $\A$ for $\Valid = \sendL \cap\, \Pref(\recL)$ };
			
			\node[initial above, initial text=, state, fill=yellow!30, accepting] [below=0.8cm of r0] (s0) {\Aname};
			\node[state, fill=yellow!30, accepting] (s0p) [left= of s0] {\Bname};
			\node[state, fill=yellow!30, accepting] (s1) [right= of s0] {\Cname};
			
			\path[->]
			(q0) edge [loop above] node[] {$!a_1\,|\,!a_2\,|\,!a_3\,|\,?a_1$~~~~~} (q0)
			(q0) edge node[] {$?a_2\,|\,?a_3$} (q1)
			(q1) edge [loop above] node[] {~~~$?a_2\,|\,?a_3$} (q1)
			;
			
			\path[->]
			(r0) edge[bend right=15, swap] node[] {$!a_1$} (r0p)
			(r0p) edge[bend right=15, swap] node[] {$!a_2$} (r0)
			(r0) edge node[] {$!a_3$} (r1)
			(r1) edge [loop right] node[] {$!a_3$} (r1)
			;
			
			\path[->]
			(s0) edge[bend right=15, swap] node[] {$?a_1$} (s0p)
			(s0p) edge[bend right=15, swap] node[] {$?a_2$} (s0)
			(s0) edge node[] {$?a_3$} (s1)
			(s1) edge [loop right] node[] {$?a_3$} (s1)
			;
			
			% ===
			
			\tikzset{node distance=1.3cm, every state/.style={circle,inner sep=0pt, minimum width=0.3cm}}
			
			\node[state, fill=yellow!30, label=center:\sname{q_0}{\Bname}{\Aname}] (q1) [below=2.8cm of a0]{};
			\node[initial above, initial text=$\fifo$, state, fill=yellow!30, label=center:\sname{q_0}{\Aname}{\Aname}] (q2) [right= of q1]{};
			\node[state, fill=yellow!30, label=center:\sname{q_0}{\Cname}{\Aname}] (q3) [right= of q2]{};
			\node[state, fill=yellow!30, label=center:\sname{q_1}{\Cname}{\Cname}] (q4) [right= of q3]{};
			\node[state, fill=yellow!30, label=center:\sname{q_0}{\Bname}{\Bname}] (q5) [below= of q1]{};
			\node[state, fill=yellow!30, label=center:\sname{q_0}{\Aname}{\Bname}] (q6) [below= of q2]{};
			\node[state, fill=yellow!30, label=center:\sname{q_0}{\Cname}{\Bname}] (q7) [below= of q3]{};
			\node[state, fill=yellow!30, label=center:\sname{q_1}{\Cname}{\Aname}] (q8) [below= of q4]{};
			
			\path[->]
			(q2) edge[bend right=15, swap] node[] {$!a_1$} (q1)
			(q1) edge[bend right=15, swap] node[] {$!a_2$} (q2)
			(q6) edge[bend right=15, swap] node[] {$!a_1$} (q5)
			(q5) edge[bend right=15, swap] node[] {$!a_2$} (q6)
			(q1) edge node[swap] {$?a_1$} (q5)
			(q2) edge node[] {$!a_3$} (q3)
			(q6) edge node[swap] {$!a_3$} (q7)
			(q2) edge node[] {$?a_1$} (q6)
			(q3) edge node[] {$?a_1$} (q7)
			(q3) edge [loop above] node[] {$!a_3$} (q3)
			(q7) edge [loop below] node[] {$!a_3$} (q7)
			(q3) edge node[] {$?a_3$} (q4)
			(q7) edge node[swap] {$?a_2$} (q8)
			(q4) edge [loop above] node[] {$?a_3$} (q4)
			(q8) edge node[] {$?a_3$} (q4)
			;
			% ========
			
			\node[state, fill=yellow!30, label=center:\sname{q_0}{\Bname}{\Aname}] (q1) [right=1.2cm of q4]{};
			\node[initial above, initial text=$\CS$, state, fill=yellow!30, label=center:\sname{q_0}{\Aname}{\Aname}] (q2) [right= of q1]{};
			\node[state, fill=yellow!30, label=center:\sname{q_0}{\Cname}{\Aname}] (q3) [right= of q2]{};
			\node[state, fill=yellow!30, label=center:\sname{q_1}{\Cname}{\Cname}] (q4) [right= of q3]{};
			\node[state, fill=yellow!30, label=center:\sname{q_0}{\Bname}{\Bname}] (q5) [below= of q1]{};
			\node[state, fill=yellow!30, label=center:\sname{q_0}{\Aname}{\Bname}] (q6) [below= of q2]{};
			\node[state, fill=yellow!30, label=center:\sname{q_0}{\Cname}{\Bname}] (q7) [below= of q3]{};
			\node[state, fill=yellow!30, label=center:\sname{q_1}{\Cname}{\Aname}] (q8) [below= of q4]{};
			
			\path[->]
			(q2) edge[bend right=15, swap] node[] {$\inc{x}$} (q1)
			(q1) edge[bend right=15, swap] node[] {$\inc{x}$} (q2)
			(q6) edge[bend right=15, swap] node[] {$\inc{x}$} (q5)
			(q5) edge[bend right=15, swap] node[] {$\inc{x}$} (q6)
			(q1) edge node[swap] {$\dec{x}$} (q5)
			(q2) edge node[] {$\inc{y}$} (q3)
			(q6) edge node[swap] {$\inc{y}$} (q7)
			(q2) edge node[] {$\dec{x}$} (q6)
			(q3) edge node[] {$\dec{x}$} (q7)
			(q3) edge [loop above] node[] {$\inc{y}$} (q3)
			(q7) edge [loop below] node[] {$\inc{y}$} (q7)
			(q3) edge node[] {\decy} (q4)
			(q7) edge node[swap] {$\dec{x}$} (q8)
			(q4) edge [loop above] node[] {\decy} (q4)
			(q8) edge node[swap] {\!\!\decy} (q4)
			;

		\end{tikzpicture}
	}
	
	\caption{For a FIFO machine $\hat\fifo$ with a single channel $\ch$ and the bounded language $\hat L_\ch = (ab)^\ast bb^\ast$ over $\bpair{ab}{b}$ (top leftmost), we construct a FIFO machine $\fifo$  (bottom left), together with $L_\ch = (a_1a_2)^\ast a_3a_3^\ast$, in normal form as the product of $h^{-1}(\hat\fifo)$ and an automaton for $\Valid$ (top right). From $\fifo$, we then obtain the counter machine $\CS$ (bottom right).} \label{fig:counter-from-fifo}
\end{figure}

Now suppose we are given a reachability query for $\hat\fifo$ in terms
of $\hat q \in \hat Q$ and a rational relation $\hat\Rel \subseteq
\prod_{\ch \in \Ch{\nch}} \hat\Sigma_\ch^\ast$.
The lemma below shows how to reduce it to a reachability
query in $\fifo$. Here,
for $\contents = (\contents_\ch)_{\ch \in \Ch{\nch}} \in \ChContents$, we define
$h(\contents) = (h_\ch(\contents_\ch))_{\ch \in \Ch{\nch}} \in \prod_{\ch \in \Ch{\nch}} \hat\Sigma_\ch^\ast$.
Note that $h^{-1}(\hat\Rel)$ is rational. Furthermore, $h: \Sigma^\ast \to \hat\Sigma^\ast$ is defined by $h(a) = h_\ch(a)$ for all $\ch \in \Ch{\nch}$ and $a \in \Sigma_\ch$,
and we extend this to $h: A_\fifo^\ast \to A_{\hat\fifo}^\ast$ in the expected manner.

\begin{lem}
	\label{lem:normal-form}
	We have 
	$(\hat q,\hat\contents) \in \reachsetL{\hat\fifo}{\hat\sendL}$ for some
	$\hat\contents \in \hat\Rel$ iff
	$((\hat q,q_\A),\contents) \in \reachsetL{\fifo}{\sendL}$ for
	some $q_\A \in Q_\A$ and $\contents \in h^{-1}(\hat\Rel)$.
\end{lem}

\begin{proof}
	Let us assume we have $(\hat q,\hat\contents) \in \reachsetL{\hat\fifo}{\hat\sendL}$ for some
	$\hat\contents \in \hat\Rel$. Hence, there exists $\hat{\runFIFO} \in \hat\sendL$ such that $(\hat{q}_0, \initcontents) \xrightarrow{\hat{\runFIFO}} (\hat q,\hat\contents)$. For channel $\ch$, let $\hat w_\ch = \projsend{\hat \runFIFO}{\ch}$.
	Since $\hat{\runFIFO} \in \hat\sendL$,
	we have $\hat w_\ch \in \hat L_\ch$ for all $\ch \in \Ch{\nch}$.
	Let $w_\ch \in L_\ch = h_\ch^{-1}(\hat L_\ch) \cap (w_{\ch,1})^\ast \ldots (w_{\ch,n_\ch})^\ast$
	such that $h_\ch(w_\ch) = \hat w_\ch$.
	There is a unique $\runFIFO \in A_\fifo^\ast$ such that $h(\runFIFO) = \hat \runFIFO$
	and $\projsend{\runFIFO}{\ch} = w_\ch$ and $\projrec{\runFIFO}{\ch} \in \Pref(w_\ch)$ for all $\ch \in \Ch{\nch}$.
	Here, $h: \Sigma^\ast \to \hat\Sigma^\ast$ is defined by $h(a) = h_\ch(a)$ for all $\ch \in \Ch{\nch}$ and $a \in \Sigma_\ch$,
	and we extend this to $h: A_\fifo^\ast \to A_{\hat\fifo}^\ast$ in the expected manner.
	Note that $\runFIFO \in  \sendL \mathrel{\cap} \Pref(\recL)$.
	Hence, we know that in the FIFO machine $h^{-1}({\hat\fifo})$, one has $(\hat{q}_0, \initcontents) \xrightarrow{\runFIFO} (\hat q,\contents)$ for some $\contents$ (by construction of $h^{-1}({\hat\fifotransition})$), and that $\runFIFO \in L(\A)$.  Therefore, since $\fifo$ is a product of the two machines, we can deduce that there is a run in $\fifo$ of the kind $((\hat q_0, q_\A^0), \initcontents) \xrightarrow{\runFIFO} (\hat q,q_\A),\contents)$, for some value of $q_\A$.
	Furthermore, by $h(\runFIFO) = \hat\runFIFO$, we have $h(\contents) = \hat\contents$.
	Hence, $\contents \in h^{-1}(\hat\Rel)$.

	\medskip
	
	Conversely, let us assume that $((\hat q,q_\A),\contents) \in \reachsetL{\fifo}{\sendL}$ for
        some $q_\A \in Q_\A$ and channel contents $\contents \in h^{-1}(\hat\Rel)$. Then, we know that there exists $\runFIFO \in \sendL$ such that $((\hat q_0, {q_\A^0}), \initcontents) \xrightarrow{\runFIFO} ((\hat q,q_\A),\contents)$.
	Let $\hat{\runFIFO} = h(\runFIFO)$.
	Since $\runFIFO \in \sendL$, we have
	$\projsend{\runFIFO}{\ch} \in L_\ch$ for all $\ch \in \Ch{\nch}$.
	In particular, $\projsend{\runFIFO}{\ch} \in h_\ch^{-1}(\hat L_\ch)$ and,
	therefore, $h_\ch(\projsend{\runFIFO}{\ch}) = \projsend{h(\runFIFO)}{\ch} \in \hat L_\ch$.
	We deduce $\hat{\runFIFO} \in \hat{\sendL}$. 
	Furthermore, we can execute $\hat{\runFIFO}$ in $\hat{\fifo}$ (by construction) to reach configuration $(\hat q,\hat\contents)$ for some $\hat\contents$.
	By $\hat{\runFIFO} = h(\runFIFO)$, we have $\hat\contents = h(\contents)$. Therefore, $\hat\contents \in \hat\Rel$.
\end{proof}

\paragraph*{Reduction of Normal Form to Counter Machine}

Henceforth, we suppose that
$\fifo = (\fifostates, \Ch{\nch}, \alphabet, \fifotransition, \init)$ and $\boundedL = (L_\ch)_{\ch \in \Ch{\nhc}}$
are in normal form, where $L_\ch$ is a bounded language over $(w_{\ch,1}, \ldots, w_{\ch,n_\ch})$.
In particular, for every letter $a \in \Sigma_\ch$, there is a
unique index $i \in \{1,\ldots,n_\ch\}$ such that $a \in \ciAlph{\ch}{i}$
where $\ciAlph{\ch}{i} = \Alphw(w_{c,i})$. We denote this index $i$ by $\indexw{\ch}{a}$.

\medskip

We build a counter machine $\CS$ such that the \IBounded
rational-reachability problem for $\fifo$ can be solved
by answering a reachability query in $\CS$, using Theorem~\ref{thm:reach-counters}.
Each run in $\CS$ will simulate a run in $\fifo$.
In particular, we want a configuration of $\CS$ to allow
us to draw conclusions about the simulated configuration in $\fifo$.
The difficulty here is that counter values are just natural numbers
and a priori store less information than channel contents with their
messages.
To overcome this, the idea is to represent each word $w_{\ch,i}$
of a tuple $(w_{\ch,1}, \ldots, w_{\ch,n_\ch})$ as a counter
$\cnt{\ch}{i}$.
Since the set of possible action sequences is
``guided'' by a bounded language, we can replace
send actions with increments and receive actions with
decrements. More precisely, $\sendaction{a}{\ch}$
becomes $(\inc{\cnt{\ch}{\indexw{\ch}{a}}},\emptyset)$,
thus incrementing the counter associated with
the unique word $w_{\ch,i}$ in which $a$ occurs.
Similarly, $\recaction{a}{\ch}$ translates to
$(\dec{\cnt{\ch}{\indexw{\ch}{a}}},Z)$ (for suitable $Z$).

This alone does not put us in a position yet where, from a counter
valuation, we can infer a unique channel contents. However, when we
additionally keep track of the last messages that have
been sent for each channel, we can
reconstruct a unique channel contents.

There is one more thing to consider here.
While the counters $\cnt{\ch}{i}$ for a given channel $\ch$
are kind of independent, the FIFO policy would not allow
us to receive a letter from $w_{\ch,j}$ while a letter from
$w_{\ch,i}$ with $i < j$ is in transit. 
Translated to the counter setting, this means that
performing $\dec{\cnt{\ch}{j}}$ should require all counters 
$\cnt{\ch}{i}$ with $i < j$ to be $0$,
so this is where zero tests come into play.
As the $L_\ch$ are bounded languages and thanks to the normal form,
however, a counter that has been tested for zero does not need to
be modified anymore.

\smallskip

We can directly implement these ideas formally
and define $\CS = (\fifostates, \Cnt, \fifotransition', \init)$ as follows
(note that $Q$ and $\init$ remain unchanged):
\begin{itemize}
	\item The set of counters is $\Cnt = \{\cnt{\ch}{i} \mid \ch \in \Ch{\nch}$ and $i \in \{1,\ldots,n_\ch\}\}$.
	
	\item For every $\sendtrans{a}{\ch}{q}{q'} \in \fifotransition$, we have
	$(q,(\inc{\cnt{\ch}{\indexw{\ch}{a}}},\emptyset),q') \in \fifotransition'$.
	
	\item For every $\rectrans{a}{\ch}{q}{q'} \in \fifotransition$, we have
	$(q,(\dec{\cnt{\ch}{\indexw{\ch}{a}}},Z),q') \in \fifotransition'$ where the set of counters to be tested for zero is
	$Z = \{\cnt{\ch}{j} \mid j < \indexw{\ch}{a}\}$.
\end{itemize}

\begin{exa}
	Figure~\ref{fig:counter-from-fifo} illustrates the construction
	of $\CS$ from a FIFO machine $\fifo$ in normal form (cf.\ Example~\ref{ex:normal-form}).
	Recall that we have one channel $\ch$ and the bounded language
	$L_\ch = (a_1a_2)^\ast a_3a_3^\ast$ over $\bpair{a_1a_2}{a_3}$.
	Thus, $\CS$ will have two counters, say $x$ for $a_1a_2$ and
	$y$ for $a_3$. Note that performing $\dec{y}$ indeed comes with
	a test of $x$ for zero.
	
	Let us first observe that it is actually important that the
	FIFO machine satisfies the trace property. Suppose that, rather than from $\fifo$,
	we constructed the counter machine directly from $h^{-1}(\hat\fifo)$. Then, configuration
	$(q_1,(1,0))$ would be reachable in the counter machine via
	$\inc{x}\inc{x}\dec{x}$, which arises from
	$\sendaction{a_1}{\ch}\sendaction{a_2}{\ch}\recaction{a_2}{\ch}$.
	However the only corresponding
	trace from $\Pref(\Valid)$ is $\sendaction{a_1}{\ch}\sendaction{a_2}{\ch}\recaction{a_1}{\ch}$,
	which in the FIFO machine $h^{-1}(\hat\fifo)$ leads to $q_0$.
	
	So consider $\fifo$ and its counter machine $\CS$.
	A channel contents $\contents \in \Sigma_\ch^\ast$ (here, we have one channel) has a natural
	counter analogue $\cvalue{\contents} =
	(|\contents|_{a_1} + |\contents|_{a_2},|\contents|_{a_3})$.
	In fact, if $(\overline{q},\contents)$ is reachable in $\fifo$, then following
	the corresponding transitions in $\CS$ will lead us to
	$(\overline{q},\cvalue{\contents})$.
	For example, $((q_0,\Cname,\Bname),a_2a_3a_3)$ is reachable in $\fifo$
	along the trace $\sendaction{a_1}{\ch}\sendaction{a_2}{\ch}\recaction{a_1}{\ch}\sendaction{a_3}{\ch}\sendaction{a_3}{\ch}$, and so is
	$((q_0,\Cname,\Bname),(1,2))$ in $\CS$ along
	$\inc{x}\inc{x}\dec{x}\inc{y}\inc{y}$ (all zero tests are empty).
	
	But how about the converse? In general, one may associate
	with a counter valuation such as $(4,0)$ several channel contents.
	Actually, both $a_1a_2a_1a_2$ and $a_2a_1a_2a_1$ seem suitable.
	However, if we know the most recent message that has been sent, say $a_1$,
	then this leaves only one option, namely $a_2a_1a_2a_1$.
	In this way, we can associate with each counter valuation
	$\cscontents$ and message $a_i \in \Sigma_\ch$
	a unique (if it exists at all) possible channel contents $\sem{\cscontents}_{a_i}$.
	Suppose that $\runCS$ is a trace in $\CS$ arising from a trace
	$\runFIFO$ in $\fifo$ whose last sent message is $a_i$.
	If $(\overline{q},\cscontents)$ is reachable in $\CS$
	via $\runCS$, then $(\overline{q},\sem{\cscontents}_{a_i})$
	is reachable in $\fifo$ via $\runFIFO$.
	For example, $\runCS = \inc{x}\inc{x}\dec{x}$ allows us
	to go to configuration $((q_0,\Aname,\Bname),(1,0))$.
	It arises from $\runFIFO = \sendaction{a_1}{\ch}\sendaction{a_2}{\ch}\recaction{a_1}{\ch}
	\in \Pref(\Valid)$, whose last sent message is $a_2$.
	We have $\sem{(1,0)}_{a_2} = a_2$. Indeed,
	$\runFIFO$ leads to $((q_0,\Aname,\Bname),a_2)$.
	\exend
\end{exa}

\paragraph*{Relation between FIFO Machine and Counter Machine}

Recall that the FIFO machine
$\fifo = (\fifostates, \Ch{\nch}, \alphabet, \fifotransition, \init)$ and $\boundedL = (L_\ch)_{\ch \in \Ch{\nhc}}$
are in normal form, where $L_\ch$ is a bounded language over $(w_{\ch,1}, \ldots, w_{\ch,n_\ch})$.
Let $\CS = (\fifostates, \Cnt, \fifotransition', \init)$ be the associated counter machine.
We will now formalize the tight
forth-and-back correspondence that allows us to solve
reachability queries in $\fifo$ in terms of
reachability queries in $\CS$.

We start with a simple observation concerning the traces of $\fifo$ and $\CS$.

\begin{lem}
	\label{lem:traceobs}
	We have $\Traces{\fifo} \subseteq \Pref(\Valid)$
	and $\Traces{\CS} \subseteq \BCounters{\CS}$.
\end{lem}

\begin{proof}
	Observe that $\Traces{\fifo} \subseteq \Traces{(\fifostates,\fifoAlpha{\fifo},\fifotransition,\init)} \subseteq \Pref(\Valid)$.
	Thus, the first property holds.

	\smallskip
	
	For the second statement, consider
	\[(\init,\initcntcontents) =
	(q_0, \cscontents_0)
	\xrightarrow{\alpha_1}_\CS
	(q_1, \cscontents_1)
	\xrightarrow{\alpha_2}_\CS
	\ldots
	\xrightarrow{\alpha_n}_\CS
	(q_n, \cscontents_n)
	\]
	and let $\tau = \alpha_1 \ldots \alpha_n$.
	Thus, $\tau \in \Traces{\CS}$.
	Suppose that we apply a zero test at position $\ell \in \{1,\ldots,n\}$, i.e.,
	$\alpha_\ell = (\dec{\cnt{\ch}{j}},Z)$ for some $(\ch,j)$,
	where $Z$ contains the counters $\cnt{\ch}{i}$ with $i < j$.
	Then, there is $k < \ell$ such that $\alpha_k = (\inc{\cnt{\ch}{j}},\emptyset)$.
	By the construction of $\CS$, we have transitions
	$(q_{k-1},\sendaction{a}{\ch},q_k)$ and
	$(q_{\ell-1},\recaction{b}{\ch},q_\ell)$ in $\fifo$
	for some $a,b \in \Sigma_{\ch,j}$. By the trace property of $\fifo$,
	none of the actions ``reachable'' from $q_\ell$ in $\fifo$
	employs a message from $\Sigma_{\ch,i}$, for all $i < j$.
	Thus, none of the actions $\alpha_m$ with $\ell \le m$ modifies a
	counter from $Z$.
	We deduce that $\runCS \in \smash{\BCounters{\CS}}$.
\end{proof}

With every channel contents $\contents \in \ChContents$ of the FIFO machine $\fifo$, we associate a counter valuation $\cvalue{\contents} = \chcontents \in \N^\Cnt$ where, for each counter $\cnt{\ch}{i}$, we let $\cscontentsp{\cnt{\ch}{i}} = \sum_{a \in \ciAlph{\ch}{i}} | \contentsp{\ch}|_a$.
Furthermore, abusing notation, we define a homomorphism $\crun{\,.\,}:  \fifoAlpha{\fifo}^\ast \to\csAlpha{\CS}^\ast$ which maps a sequence of actions of $\fifo$ to a sequence of actions of $\CS$. It is defined by $ \crun{\sendaction{a}{\ch}} = (\inc{\cnt{\ch}{\indexw{\ch}{a}}}, \emptyset)$ and $\crun{\recaction{a}{\ch}} = (\dec{\cnt{\ch}{\indexw{\ch}{a}}},Z)$ where $Z = \{\cnt{\ch}{j} \mid j < \indexw{\ch}{a}\}$.

Conversely, we will associate, with counter values and traces of $\CS$
the corresponding objects in the FIFO machine.
Because of the inherent ambiguity, this is, however, less straightforward.
First, we define a partial mapping
$\frun{\,.\,} : \csAlpha{\CS}^\ast \to \fifoAlpha{\fifo}^\ast$ (that is not a homomorphism).
For $\runCS \in \csAlpha{\CS}^\ast$, we let $\frun{\runCS}$ be the unique (if it exists) word
$\runFIFO \in \Pref(\Valid)$ such that $\crun{\runFIFO} = \runCS$.

\smallskip

Next, we associate with a counter valuation a corresponding channel contents.
As explained above, there is no unique choice unless we make an assumption on
the last messages that have been sent.
For $\ch \in \Ch{\nch}$, we set $\Sigmabot_\ch = \Sigma_\ch \uplus \{\bot\}$.
Let $a \in \Sigmabot_\ch$ and $w \in \Sigma_\ch^\ast$.
We say that $a$ is \emph{good} for
$w$ if $w \in \Inf(L_\ch)$ and either
$w = \varepsilon$
or
$w = u.a$ for some $u \in \Sigma_\ch^\ast$.
Intuitively, it may be possible to obtain
contents $w$ in channel $\ch$ when $a$ is the last message sent
(no message was sent yet through $\ch$ if $a=\bot$).
Note that the set of words $w \in \Sigma_\ch^\ast$ such
that $a$ is good for $w$ is a regular language.
Moreover, with $\contents \in \ChContents$,
we associate the finite set
$\pairs{\contents} \subseteq \prod_{\ch \in \Ch{\nch}} \Sigmabot_\ch$
of tuples $\veca = (\veca_\ch)_{\ch \in \Ch{\nch}}$ such that,
for all $c \in \Ch{\ch}$, $\veca_\ch$ is good for $\contents_\ch$.

Let
$\cscontents \in \N^\Cnt$ and $\veca \in \prod_{\ch \in \Ch{\nch}} \Sigmabot_\ch$.
Abusing notation, we will
associate with $\cscontents$ and $\veca$ 
the channel contents
$\fvalue{\cscontents}{\veca} \in \ChContents$ (if it exists).
We let 
$\fvalue{\cscontents}{\veca} = \contents$ if
$\cvalue{\contents} = \cscontents$ and
$\veca \in \pairs{\contents}$.
There is at most one such $\contents$ so that this is well-defined.
Note that $\cvalue{\fvalue{\cscontents}{\veca}} = \cscontents$.

\begin{exa}
	If we have one channel $\ch$ and our bounded language is
	$L_\ch=(a_1a_2a_3)^\ast (a_4)^\ast$, then 
	$\fvalue{(4,0)}{a_2} = a_2a_3a_1a_2$
	and 
	$\fvalue{(2,1)}{a_4} = a_2a_3a_4$, whereas
	$\fvalue{(3,1)}{a_3}$ is undefined.
	Moreover, 
	$\pairs{a_2a_3} = \{a_3\}$ and
	$\pairs{\varepsilon} = \{a_1, a_2, a_3, a_4, \bot\}$.
	\exend
\end{exa}

Given $\cscontents$ and $\veca$, 
we can easily compute  $\fvalue{\cscontents}{\veca}$ since there are only finitely many words $\contents$ for a given $\cscontents$ such that $\cvalue{\contents} = \cscontents$. Furthermore, we can also compute $ \pairs{\contents}$ for a given $\contents$ as we have finitely many possibilities of $\veca$.

\newcommand{\Lcontents}[1]{L_{#1}}

Finally, for $\veca \in \prod_{\ch \in \Ch{\nch}} \Sigmabot_\ch$,
we let $\abL{\veca} \subseteq \fifoAlpha{\fifo}^\ast$ be the set of
words $\runFIFO$ such that, for all $\ch \in \Ch{\nch}$, $\veca_\ch$ 
is the last message sent to $\ch$ in $\runFIFO$ (no message was sent if $\veca_\ch = \bot$).
We are now ready to state that runs in the FIFO machine
are faithfully simulated by runs in the counter machine
(the proof is by induction on the length of the trace):

\begin{prop}
	\label{prop:fifo-cm}
	Let $\runFIFO  \in \fifoAlpha{\fifo}^\ast$. 
	For all $(q, \contents) \in \tscomp{\tsstates}{\fifo}$ and $\veca \in \prod_{\ch \in \Ch{\nch}} \Sigmabot_\ch$
	such that
	$\runFIFO \in \abL{\veca}$, we have:
	$
	(\init,\initcontents) \xrightarrow{ \runFIFO }_\fifo (q, \contents)
	\implies
	\bigl(
	(\init,\initcntcontents) \xrightarrow{\crun{\runFIFO}}_\CS (q, \cvalue{\contents})	
	\text{ and } \veca \in \pairs{\contents}
	\bigr)\,.
	$
\end{prop}

\begin{proof}\label{app:fifo-cm-proof}
	We will prove
	the statement by induction on the length of $\runFIFO$.
	In the base case, $|\runFIFO|=0$.
	The only value of $\runFIFO$ such that $|\runFIFO|=0$ is $\runFIFO =\varepsilon$. Furthermore, $\varepsilon \in \Pref(\Valid) \cap \abL{\veca}$ where $\veca = (\bot)^\Ch{\nch}$. The initial configuration $(\init, \initcontents)  \in \tscomp{\tsstates}{\fifo}$ is the only configuration reachable by $\varepsilon$.
	In $\TSof{\CS}$, the only configuration reachable via $\crun{\varepsilon} = \varepsilon$ is $(\init, \initcntcontents)$ = $(q, \crun{\initcontents})$, and $\crun{\varepsilon} \in \BCounters{\CS}$. Finally, we also see that
	$\veca \in \pairs{\initcontents}$. 
	Therefore,
	the base case is valid.
	
	\medskip
	
	We suppose that
	the statement is true for all $\runFIFO \in \fifoAlpha{\fifo}^\ast$ such that $|\runFIFO| = n$. We will now show that
	it is true for  $\runFIFO' \in \fifoAlpha{\fifo}^\ast$ where $|\runFIFO'|=n+1$.
	Let $\veca' \in \prod_{\ch \in \Ch{\nch}} \Sigmabot_\ch$ and suppose $\runFIFO' \in \Pref(\Valid) \cap \abL{\veca'}$.	

	We write $\runFIFO'=\runFIFO.\FIFOaction$ such that $\runFIFO \in \Pref(\Valid) \cap \abL{\veca}$ for some $\veca \in \prod_{\ch \in \Ch{\nch}} \Sigmabot_\ch$, $|\runFIFO| = n$, and $\FIFOaction \in \fifoAlpha{\fifo}$. There exists such a $\runFIFO$ since the set $\Pref(\Valid)$ is prefix-closed.
	
	Let $(q', \contents') \in \tscomp{\tsstates}{\fifo}$ such that $(\init,\initcontents) \xrightarrow{ \runFIFO' }_\fifo (q', \contents')$. Then, there is
	$(q, \contents)  \in \tscomp{\tsstates}{\fifo}$ such that $(\init,\initcontents) \xrightarrow{ \runFIFO }_\fifo (q, \contents) \xrightarrow{\FIFOaction}_\fifo (q', \contents')$.
	Hence, there exists  $t = (q, \FIFOaction, q') \in \fifotransition$ in the FIFO machine.
	
	\subsection*{Case (1):}
	Suppose $\FIFOaction = \sendaction{a}{\ch}$, for some $\ch \in \Ch{\nch}$ and $a \in \Sigma_\ch$.
	Hence, we have $\contentspp{\ch} = \contentsp{\ch}.a$,
	and $\contentspp{\chb} = \contentsp{\chb}$ for all $\chb \in \Ch{\nch} \setminus \{\ch\}$.
	Moreover, since $\runFIFO' = \runFIFO.\FIFOaction$ and $\FIFOaction = \sendaction{a}{\ch}$, we can deduce that 
	$\veca'_\ch = a$ and $\veca'_d = \veca_d$ for all $d \neq \ch$. 	 This is because the only change in the last sent letters between $\runFIFO$ and $\runFIFO'$ is in the channel $\ch$.
	
	By construction of $\CS$, we know that there is a transition $t' = (q,\inc{\cnt{\ch}{\indexw{\ch}{a}}},\emptyset,q') \in \fifotransition'$ in $\CS$, and by the definition of $\crun{\,.\,}$, we also have $\crun{\FIFOaction} = (\inc{\cnt{\ch}{\indexw{\ch}{a}}},\emptyset)$.
	By induction hypothesis, we have $(\init,\initcntcontents) \xrightarrow{\crun{\runFIFO}} (q, \cscontents)$ where $\cscontents = \cvalue{\contents}$. Since $\crun{\beta}$ increases a counter, we have $(q,\cscontents) \xrightarrow{\crun{\FIFOaction}} (q', \cscontents')$ for the counter valuation $\cscontents'$
	such that $\cscontentspp{\counter} = \cscontentsp{\counter} + 1$ for $\counter =  \cnt{\ch}{\indexw{\ch}{a}}$ and
	$\cscontentspp{\counterb} = \cscontentsp{\counterb}$ for all $\counterb \in \Cnt \setminus \{\counter\}$. Hence, $\cvalue{\contents'} = \cscontents'$ and we have $(\init,\initcntcontents) \xrightarrow{\crun{\runFIFO.\FIFOaction}} (q', \cvalue{\contents'})$.
	Note that, by Lemma~\ref{lem:traceobs}, we have $\crun{\runFIFO.\FIFOaction} \in \BCounters{\CS}$.

	From the induction hypothesis, we know that $\veca \in \pairs{\contents}$. In order to show that  $\veca' \in \pairs{\contents'}$, we only need to address the case of the channel $\ch$, since the values of $\veca_\chb, \contents_\chb$ remain unchanged for all $\chb \neq \ch$.
	We know that $\contents'_\ch = \contents_\ch.\veca'_\ch$.
	By induction hypothesis, $\contents_\ch \in \Inf(L_\ch)$.
	Since $\runFIFO' \in \Pref(\sendL)$, we also have $\contents_\ch.\veca'_\ch \in \Inf(L_\ch)$.
	Hence, $\veca' \in \pairs{\contents'}$.

	\subsection*{Case (2):}
	
	Suppose $\FIFOaction = \recaction{a}{\ch}$, for some $\ch \in \Ch{\nch}$ and $a \in \Sigma_\ch$.
	We have $\contentsp{\ch} = a . \contentspp{\ch}$,
	and $\contentspp{\chb} = \contentsp{\chb}$ for all $\chb \in \Ch{\nch} \setminus \{\ch\}$.
	Since $\runFIFO' = \runFIFO.\FIFOaction$ and $\FIFOaction = \recaction{a}{\ch}$, we can deduce that $\veca' = \veca$. This is because there is no change in the last letter sent between $\runFIFO$ and $\runFIFO'$.
	
	By construction of $\CS$, we know that there is a transition $t' = (q,\dec{\cnt{\ch}{\indexw{\ch}{a}}},Z,q') \in \fifotransition'$ in $\CS$ where $Z =  \{(\cnt{\ch}{j}) \mid j < \indexw{\ch}{a}\}$. By the definition of $\crun{\,.\,}$, we also have $\crun{\FIFOaction} = (\dec{\cnt{\ch}{\indexw{\ch}{a}}},Z)$. By the induction hypothesis, we have  $(\init,\initcntcontents) \xrightarrow{\crun{\runFIFO}} (q, \cvalue{\contents})$. Furthermore, recall that $\contentsp{\ch} = a . \contentspp{\ch}$. This implies that $a$ is at the head of the channel $\ch$ in the configuration $(q, \contents)$. Hence, all the letters $b$ such that $\indexw{\ch}{b} < \indexw{\ch}{a}$ are not present in the channel. Therefore, in the configuration $(q, \cvalue{\contents})$, all the counters $\cnt{\ch}{\indexw{\ch}{b}}$ are equal to zero for $\indexw{\ch}{b} < \indexw{\ch}{a}$. Hence, we can execute the transition $t'$ from $(q, \cvalue{\contents})$, and we have $(\init,\initcntcontents) \xrightarrow{\crun{\runFIFO}} (q, \cvalue{\contents})  \xrightarrow{\crun{\FIFOaction}} (q', \cscontents')$ for some counter valuation $\cscontents'$. By Lemma~\ref{lem:traceobs}, we have $\crun{\runFIFO.\FIFOaction} \in \smash{\BCounters{\CS}}$.
	We write $\cscontents = \cvalue{\contents}$, and from the counter machine transition relation, we know that $\cscontentspp{\counter} = \cscontentsp{\counter} - 1$ for $\counter =  \cnt{\ch}{\indexw{\ch}{a}}$ and
	$\cscontentspp{\counterb} = \cscontentsp{\counterb}$ for all $\counterb \in \Cnt \setminus \{\counter\}$. Hence, $\cvalue{\contents'} = \cscontents'$.

	From the induction hypothesis, we know that $\veca \in \pairs{\contents}$. In order to show that  $\veca' \in \pairs{\contents'}$, we only need to address the case of the channel $\ch$, since the values of $\veca_\chb, \contents_\chb$ remain unchanged for all $\chb \neq \ch$. 
	We know from the  induction hypothesis that
	$\contents_\ch \in  \Inf(L_\ch)$. Hence, we can immediately deduce that $\contents'_\ch \in \Suf(\contents_\ch) \subseteq \Inf(L_\ch)$.
	
	Furthermore, we know that $\contents_\ch = u. \veca'_\ch$ for some $u \in \Sigma_\ch^\ast$. If $u = \varepsilon$, then we can deduce that $\contents'_\ch = \varepsilon$. If $u \neq \varepsilon$, then we have $\contents'_\ch = u'.\veca'_\ch$ such that $a.u' = u$.
	Hence,
	$\veca' \in \pairs{\contents'}$.
\end{proof}

Conversely, we can show that runs of the counter machine
can be retrieved in the FIFO machine (again, the proof proceeds by induction
on the length of the trace):

\begin{prop}
	\label{prop:cm-fifo}
	Let $\runCS  \in \csAlpha{\CS}^\ast$. For all  $(q, \cscontents) \in \tscomp{\tsstates}{\CS}$ and
	$\veca \in \prod_{\ch \in \Ch{\nch}} \Sigmabot_\ch$
	such that $\runCS \in \crun{\Pref(\Valid) \cap \abL{\veca}}$, we have:
	$
	(\init,\initcntcontents) \xrightarrow{\runCS}_\CS (q, \cscontents)
	\implies
	(\init,\initcontents) \xrightarrow{ \frun{\runCS} }_\fifo (q, \fvalue{\cscontents}{\veca})\,.
	$
\end{prop}

\begin{proof}\label{app:cm-fifo-proof}
	We proceed by induction on the length of $\runCS$. In the base case, $|\runCS|=0$. The only value of $\runCS$ such that $|\runCS|=0$ is $\runCS =\varepsilon$. Furthermore, $\varepsilon \in \crun{\Pref(\Valid) \cap \abL{\veca}}$ where $\veca_\ch = \bot$ for all $\ch \in \Ch{\nch}$. The only configuration reachable via $\varepsilon$ is $(\init,\initcntcontents)$. In the FIFO machine, the configuration $(\init,\initcontents)$ is the only configuration reachable via $\frun{\varepsilon} = \varepsilon$. We know that the initial contents is $\initcontents = \fvalue{\initcntcontents}{\veca}$. Hence, the base case is valid.
	
	Let us suppose that
	the statement holds for $\runCS \in \csAlpha{\CS}^\ast$ where $|\runCS|=n$. We show that it is true for $\runCS'$ with $|\runCS'| = n+1$. Let $\veca'$ such that $\runCS' \in \crun{\Pref(\Valid) \cap \abL{\veca'}}$. Then we can write $\runCS' = \runCS.\CSaction$ for some  $\runCS \in \fifoAlpha{\CS}^\ast$ and $\CSaction \in \fifoAlpha{\CS}$.
	There exists $\runFIFO' \in \Pref(\Valid) \cap \abL{\veca'}$ such that $\crun{\runFIFO'} = \runCS'$.
	Furthermore, since $\crun{\,.\,}$ is a homomorphism, we can express $\runFIFO' = \runFIFO.\FIFOaction$ for some $\runFIFO \in \fifoAlpha{\fifo}^\ast$ and $\FIFOaction \in \fifoAlpha{\fifo}$ where $\crun{\FIFOaction} = \CSaction$ and $\crun{\runFIFO} = \runCS$. Since $\Pref(\Valid)$ is prefix-closed, we have
	$\sigma \in \Pref(\Valid) \cap \abL{\veca}$ for some $\veca \in \prod_{\ch \in \Ch{\nch}} \Sigmabot_\ch$.
	Therefore, $\runCS \in \crun{\Pref(\Valid) \cap \abL{\veca}}$. 
	
	Let $(\init,\initcntcontents) \xrightarrow{ \runCS' }_\CS (q', \cscontents')$.
	Note that, by Lemma~\ref{lem:traceobs}, we have $\runCS' \in \BCounters{\CS}$.
	We will prove that $(\init,\initcontents) \xrightarrow{ \frun{\runCS'} }_\fifo (q, \contents')$ where $\contents' = \fvalue{\cscontents'}{\veca'}$. Since $\runCS' = \runCS.\CSaction$, there is $(q, \cscontents) \in \tscomp{\tsstates}{\CS}$ such that $(\init, \initcntcontents) \xrightarrow{\runCS} (q, \cscontents) \xrightarrow{\CSaction} (q', \cscontents')$. Hence, there exists a transition $t = (q, \CSaction, q') \in \fifotransition'$.
	
	By the induction hypothesis, we know that
	$(\init,\initcontents) \xrightarrow{ \smash{\frun{\runCS}} }_\fifo (q, \contents)$ where $	\contents = \fvalue{\cscontents}{\veca}$.

	\subsection*{Case (1):}
	Suppose  $\CSaction =(\inc{\cnt{\ch}{i}},\emptyset)$ for $\ch \in \Ch{\nch}$ and $i \in \{1,\ldots,n_\ch\}$.
	We have $\cscontentspp{\counter} = \cscontentsp{\counter} + 1$ for $\counter = \cnt{\ch}{i}$ and
	$\cscontentspp{\counterb} = \cscontentsp{\counterb}$ and for all $\counterb \in \Cnt \setminus \{\counter\}$.
	
	By construction of $\CS$, we know that there is a transition $t' = (q, \gamma, q') \in \fifotransition$ in $\fifo$ with $\gamma = \sendaction{a}{\ch}$ for some $a \in \Sigma_\ch$ such that $\indexw{\ch}{a} = i$.
	Thanks to the trace property (Definition~\ref{def:normal-form} (2.)), we have
	$\beta = \gamma$.
	Let $\contents'$ be given by $\contentspp{\ch} = \contentsp{\ch}.a$
	and $\contentspp{\chb} = \contentsp{\chb}$ for all $\chb \in \Ch{\nch} \setminus \{\ch\}$.
	Then, $(q,\contents) \xrightarrow{ \beta }_\fifo (q', \contents')$.
	Furthermore, since $\indexw{\ch}{a} = i$ and $\cvalue{\contents} = \cscontents$, we can deduce that  $\cvalue{\contents'} = \cscontents'$. Moreover, recall that $\crun{\runFIFO.\FIFOaction} = \runCS.\CSaction$, hence, $\runFIFO' = \runFIFO.\FIFOaction = \frun{\runCS'}$. 
	
	Since $\runFIFO' = \runFIFO.\FIFOaction$ and $\FIFOaction = \sendaction{a}{\ch}$, we can deduce 
	that $\veca'_\ch = a$ and $\veca'_d = \veca_d$ for all $d \neq \ch$. 	 This is because the only change in the last sent letters between $\runFIFO$ and $\runFIFO'$ is in the channel $\ch$.
	
	We recall that $\cvalue{\contents'} = \cscontents'$. From the induction hypothesis, we know that $\fvalue{\cscontents}{\veca}{} = \contents$. Hence, in order to show that  $\fvalue{\cscontents'}{\veca'}{} = \contents'$, we only need to
	address the case of the channel $\ch$, since the values of $\veca_\chb, \contents_\chb$ remain unchanged for all $\chb \neq \ch$.
	We know that $\contents'_\ch = \contents_\ch.\veca'_\ch$.
	Since $\runFIFO' \in \Pref(\sendL)$, we have $\contents_\ch.\veca'_\ch \in \Inf(L_\ch)$.
	Therefore, $\contents' = \fvalue{\cscontents'}{\veca'}{}$.

	\subsection*{Case (2):}
	Suppose $\CSaction = (\dec{\cnt{\ch}{i}},Z)$,  for $\ch \in \Ch{\nch}$, $i \in \{1,\ldots,n_\ch\}$ and $Z = \{\cnt{\ch}{j} \mid j <i\}$.
	We have $\cscontentspp{\counter} = \cscontentsp{\counter} - 1$ for $\counter = \cnt{\ch}{i}$ and
	$\cscontentspp{\counterb} = \cscontentsp{\counterb}$ for all $\counterb \in \Cnt \setminus \{\counter\}$.
	
	By construction of $\CS$, we know that there is a transition $t' = (q, \gamma, q') \in \fifotransition$ in $\fifo$ with $\gamma = \recaction{a}{\ch}$ for some $a \in \Sigma_\ch$ such that $\indexw{\ch}{a} = i$.
	Again, by the trace property (Definition~\ref{def:normal-form} (2.)), we get
	$\beta = \gamma$.
	In order to execute $t'$ from $(q, \contents)$, it is necessary that we have $\contents_\ch = a.u$ for some word $u$.

	Since $\runFIFO.\FIFOaction \in \Pref(\recL)$ and  $\projrec{\runFIFO}{\ch}.\contentsp{\ch} = \projsend{\runFIFO}{\ch}$, we can deduce that $\contentsp{\ch} = a.u$ for some word $u$. 	Hence, the transition $t'$ can be executed to reach a configuration $(q', \contents')$ such that $\contentsp{\ch} = a \cdot \contentspp{\ch}$,
	and $\contentspp{\chb} = \contentsp{\chb}$ for all $\chb \in \Ch{\nch} \setminus \{\ch\}$.  Furthermore, since $\indexw{\ch}{a} = i$ and $\cvalue{\contents} = \cscontents$, we can deduce that  $\cvalue{\contents'} = \cscontents'$. Moreover, recall that $\crun{\runFIFO.\FIFOaction} = \runCS.\CSaction$, hence, $\runFIFO' = \runFIFO.\FIFOaction = \frun{\runCS'}$. 
	
	Since $\runFIFO' = \runFIFO.\FIFOaction$ and $\FIFOaction = \recaction{a}{\ch}$, we can deduce that $\veca' = \veca$. This is because no letters are sent
	between $\runFIFO$ and $\runFIFO'$.
	
	We also know from the  induction hypothesis that $\contents = \fvalue{\cscontents}{\veca}{}$. Also recall that $\cvalue{\contents'} = \cscontents'$. In order to show that  $\contents' = \fvalue{\cscontents'}{\veca'}{}$, we only need to address the case of the channel $\ch$, since the values of $\veca_\chb, \contents_\chb$ remain unchanged for all $\chb \neq \ch$. 
	
	We know from the induction hypothesis that $\contents_\ch$ is contained in $\Inf(L_\ch)$ and, thus, so is $\contents_\ch'$.		
	Furthermore, we know that $\contents_\ch = u. \veca'_\ch$ for some $u \in \Sigma_\ch^\ast$. If $u = \varepsilon$, then we can deduce that $\contents'_\ch = \varepsilon$. On the other hand, if $u \neq \varepsilon$, then we know that $\contents'_\ch = u'.\veca'_\ch$ such that $a.u' = u$. We also recall that $\cvalue{\contents'} = \cscontents'$. Hence, $\contents' = \fvalue{\cscontents'}{\veca'}{}$.
	\qedhere
\end{proof}

From Propositions~\ref{prop:fifo-cm} and \ref{prop:cm-fifo} and Lemma~\ref{lem:traceobs}, we obtain the following corollary.

\begin{cor}\label{reachprop}
	For all $(q, \contents) \in \tscomp{\tsstates}{\fifo}$,
	we have:
	$
	(q, \contents) \in \reachsetL{\fifo}{\sendL}
	\Longleftrightarrow
	(q, \cvalue{\contents}) \in
	\textstyle \reachsetL{\CS}{\BCounters{\CS} \cap \cvalue{\Valid
			\cap \bigcup_{\veca \in \pairs{\contents}} \abL{\veca}}}\,.
	$
\end{cor}

From Theorem~\ref{thm:reach-counters}, we know that verifying whether $(q, \cvalue{\contents}) \in \reachsetL{\CS}{\BCounters{\CS} \cap L} $ where $L=  \crun{\Valid
	\cap \bigcup_{\veca \in \pairs{\contents}} \abL{\veca}}$ is decidable. Hence, we can
already deduce decidability of the (configuration-)reachability problem.
In fact, using Propositions~\ref{prop:fifo-cm} and \ref{prop:cm-fifo},
we can solve the more general \IBounded rational-reachability problem.
For this, it is actually enough to check, in the counter machine, the reachability of a counter value that belongs to a semi-linear set. For $\veca \in \prod_{\ch \in \Ch{\nch}} \Sigmabot_\ch$ and
a rational relation $\Rel \subseteq \ChContents$,
let $\TestL{\veca}{} = \{\cscontents \in \N^\Cnt \mid \fvalue{\cscontents}{\veca}{} \in \textscrb{R}\}$.

\begin{lem}
	\label{lem:semilinear}
	The set $\TestL{\veca}{}$ is effectively semi-linear.
\end{lem}

\begin{proof}
	\label{app:semilinear-proof}
	For $\ch \in \Ch{\nch}$, let $\mathcal{G}_\ch$ be the set of words $w \in \Sigma_\ch^\ast$
	such that $\veca_\ch$ is good for $w$.
	Moreover, let $\Good = \prod_{\ch \in \Ch{\nch}} \mathcal{G}_\ch$.

	As $\mathcal{G}_\ch$ is regular for every $\ch \in \Ch{\nch}$, the relation
	$\Good$ is recognizable. As the intersection of a rational
	and a recognizable relation is rational, we have that
	$\Rel \cap \Good$ is rational. It follows that
	$\Parikhimg{\Rel \cap \Good}$ is semi-linear.
	From the definitions, we obtain
	% ===
	% ===
	\[\begin{array}{rlll}
		
		% ===
		\fvalue{\cscontents}{\veca}{} \in \textscrb{R} & \Longleftrightarrow &
		\exists \contents \in \Rel \cap \Good:
		&
		\forall \cnt{\ch}{i} \in \Cnt:
		\cscontents_{\cnt{\ch}{i}} = \sum_{a \in \ciAlph{c}{i}} |\contents_\ch|_a
		\\
		% ===
		& \Longleftrightarrow &
		\exists \pi \in \Parikhimg{\Rel \cap \Good}:\!\!\!
		&
		\forall \cnt{\ch}{i} \in \Cnt:
		\cscontents_{\cnt{\ch}{i}} = \sum_{a \in \ciAlph{c}{i}} \pi_a
	\end{array}\]

	Thus,
	\[
	\TestL{\veca}{} = \{\cscontents \in \N^\Cnt \mid \fvalue{\cscontents}{\veca}{} \in \textscrb{R}\}
	=
	\textstyle \bigl\{\bigl(\sum_{a \in \ciAlph{\ch}{i}} \pi_a\bigr)_{\cnt{\ch}{i} \in \Cnt} \mid \pi \in \Parikhimg{\Rel \cap \Good}\bigr\}\,.
	\]
	Let $\varphi((X_a)_{a \in \Sigma})$ be a Presburger formula defining $\Parikhimg{\Rel \cap \Good}$. Then, by the above equivalence, the following Presburger formula defines $\TestL{\veca}{\vecb}$:
	\[\psi_\ch((Z_y)_{y \in \Cnt}) =
	\exists (X_a)_{a \in \Sigma}:
	\Bigl(
	\varphi((X_a)_{a \in \Sigma})
	\wedge
	\bigwedge_{y \in \Cnt}
	Z_{y} = \sum_{a \in \cAlph{y}} X_a
	\Bigl)
	\]
	where, for $y = \cnt{\ch}{i} \in \Cnt$, we let $\cAlph{y} = \ciAlph{\ch}{i}$.
	It follows that $\TestL{\veca}{\vecb}$ is effectively semi-linear.
\end{proof}

Using this property, we finally reduce the \IBounded rational-reachability problem
to a reachability problem in counter machines:

\begin{cor}
	\label{corr:reacheq}
	For every $q \in \fifostates$, we have:
	$(q, \contents) \in \reachsetL{\fifo}{\sendL}
	\text{ for some $\contents \in \textscrb{R}$}
	\Longleftrightarrow 
	(q, \cscontents) \in
	\textstyle \reachsetL{\CS}{\BCounters{\CS} \cap \crun{\Valid
			\cap  \abL{\veca}{}}}
	\textup{ for some } \veca \in \prod_{\ch \in \Ch{\nch}} \Sigmabot_\ch \text{ and } \cscontents \in \TestL{\veca}{}\,.$
\end{cor}

\begin{proof}
	\label{app:reach-eq-proof}
	Suppose $(q, \contents) \in \reachsetL{\fifo}{\sendL}$ with
	$\contents \in \textscrb{R}$.
	There are $\veca \in \prod_{\ch \in \Ch{\nch}} \Sigmabot_\ch$ and $\runFIFO \in \Valid \cap \abL{\veca}{}$ such that
	$(\init,\initcontents) \xrightarrow{ \runFIFO }_\fifo (q, \contents)$.
	By Proposition~\ref{prop:fifo-cm} and Lemma~\ref{lem:traceobs}, we have
	$(\init,\initcntcontents) \xrightarrow{\crun{\runFIFO}}_\CS (q, \crun{\contents})$
	and $\crun{\runFIFO} \in \BCounters{\CS}$ and $\veca \in \pairs{\contents}$.
	By definition, the latter implies $\contents = \fvalue{\cvalue{\contents}}{\veca}{}$
	and hence $\cvalue{\contents} \in \TestL{\veca}{}$.
	
	\smallskip
	
	Conversely, suppose we have
	$(\init,\initcntcontents) \xrightarrow{\crun{\runFIFO}}_\CS (q, \chcontents)$
	where $\runFIFO \in \Valid
	\cap  \abL{\veca}{}$, $\crun{\runFIFO} \in \BCounters{\CS}$,
	and $\cscontents \in \TestL{\veca}{}$.
	By Proposition~\ref{prop:cm-fifo}, we get
	$(\init,\initcontents) \xrightarrow{ \frun{\crun{\runFIFO}} }_\fifo (q, \fvalue{\cscontents}{\veca}{})$.
	Note that $\frun{\crun{\runFIFO}} = \runFIFO \in \sendL$.
	Moreover, $\cscontents \in \TestL{\veca}{}$ implies $\fvalue{\cscontents}{\veca}{} \in \textscrb{R}$, which concludes the proof.
\end{proof}

By Theorem~\ref{thm:reach-counters}, we can now deduce
Theorem~\ref{thm:general-I-bounded-reach}, i.e.,
decidability of \IBounded rational-reachability.

\section{Reachability and Deadlock}\label{sec:variants}

We now address some other commonly studied reachability problems,
which, as it turns out, can be reduced to the \IBounded
rational-reachability problem studied in the previous section.

A configuration $(q,\contents)$ of a FIFO machine $\fifo$
is a \emph{deadlock} if there is no $(q',\contents')$ such that $(q,\contents) \tscomp{\tstrans}{\fifo} (q',\contents')$.

\begin{defi}[\IBounded decision problems]
	Given a FIFO machine $\fifo = (\fifostates, \Ch{\nch}, \alphabet, \fifotransition, \init)$, a control-state $q \in Q$, a configuration $s \in \tscomp{\tsstates}{\fifo}$, and a tuple
	$\boundedL = (L_\ch)_{\ch \in \Ch{\nch}}$ of non-empty regular bounded languages $L_\ch \subseteq \Sigma_\ch^\ast$.
	\begin{itemize}
		\item \emph{\IBounded reachability}: Do we have $s \in \reachsetL{\fifo}{\sendL}$\,?
		\item \emph{\IBounded control-state reachability}: Do we have
		$(q,\contents) \in \reachsetL{\fifo}{\sendL}$ for some $\contents$\,?
		\item \emph{\IBounded deadlock}: Does $\reachsetL{\fifo}{\sendL}$ contain a deadlock\,?

	\end{itemize}
\end{defi}

In \cite{DBLP:conf/concur/FinkelP19}, it was shown that reachability reduces to control-state reachability for flat FIFO machines but the converse is not true. However, using the same reductions as in \cite{DBLP:conf/concur/FinkelP19}, we obtain
the following results: 

\begin{prop}
	\label{prop:reach-equiv-csr}
	\IBounded  reachability is
	\begin{enumerate}
		\item[\textup{(a)}] recursively equivalent to  \IBounded control-state reachability for FIFO machines, and
		\item[\textup{(b)}] recursively reducible to \IBounded deadlock for FIFO machines.
	\end{enumerate}
\end{prop}

\begin{proof}\label{app:reach-eq-csr-proof}
	We show both parts of the lemma.
	\subsection*{Part (a):}
	Let us consider a FIFO machine $\fifo  = (\fifostates, \Ch{\nch}, \alphabet, \fifotransition, \init)$
	with $\Ch{\nch} = \{1,\ldots,\nch\}$, a control-state $q$, a configuration $(q,\contents)$, and a tuple $\boundedL= (L_\ch)_{\ch \in \Ch{\nch}}$ of non-empty regular bounded languages $L_\ch \subseteq \alphabet_\ch^\ast$.
	Suppose $\contents = (\contents_\ch)_{\ch \in \Ch{\nch}}$
	with $\contents_\ch = \contents_\ch^1 \ldots \contents_\ch^{n_\ch}$.

\medskip

	We first reduce \IBounded reachability to \IBounded control-state reachability. Another machine $\fifo_{(q,\contents)}  = (\fifostates', \Ch{\nch}, \alphabet', \fifotransition', \init)$ is constructed as follows.
	We set $\fifostates' = \fifostates \cup \{q_{\text{end}}\}$ such that $q_{\text{end}} \notin \fifostates$, and $\alphabet' = \alphabet \cup \{\$\}$ such that $\$ \not\in \alphabet$.
	Moreover, a ``path'' $q \xrightarrow{\runFIFO} q_{\text{end}}$ with $\runFIFO = \runFIFO_1\ldots \runFIFO_\nch$ is added from the control state $q$ as follows. For $\ch \in \Ch{\nch}$, we have
	\[
	\runFIFO_\ch =
	\begin{cases}
		\sendaction{\$}{c} \recaction{\contents_\ch^1}{c} \ldots \recaction{\contents_\ch^{n_\ch}}{c} \recaction{\$}{c} & \text{if } |\contents_\ch| > 0,\\
		\sendaction{\$}{c} \recaction{\$}{c} & \text{otherwise. }
	\end{cases}\] 

Then, $(q,\contents) \in \reachsetL{\fifo}{\sendL}$ iff $(q_{\text{end}},\contents') \in \reachsetL{\fifo_{(q,\contents)}}{\sendL'}$ for \emph{some} $\contents'$, where $\boundedLp= (L_\ch.\contents_\ch.\$)_{\ch \in \Ch{\nch}}$. Furthermore, $\boundedLp$ is bounded if $\boundedL$ is bounded, since concatenation of a finite word with a bounded language results in a bounded language. Therefore, \IBounded reachability reduces to \IBounded control-state reachability for FIFO machines.

\medskip

Conversely, in order to show that \IBounded control-state reachability is reducible to \IBounded reachability, we construct $\fifo_{q}$ as follows. To $\fifo$, we add $|\Sigma| \times \nch$ self-loops from and to the control state $q$ as follows: $q \xrightarrow{\ch?a} q$ for all $\ch \in \Ch{\nch}$ and $a \in \alphabet_\ch$. 
	
	The control-state $q$ is reachable in $\fifo$ iff there exists $\contents$ such that $(q,\contents)$ is reachable in $\fifo$ iff $(q,\initcontents)$ is reachable in $\fifo_{q}$. Furthermore, consider $\runFIFO \in  \Pref(\sendL)$ such that $ (\init,\initcontents) \xrightarrow{ \runFIFO }_\fifo (q, \contents)$ for some channel contents $\contents$. Let us append to $\runFIFO$ a sequence of actions $\runFIFO' = \runFIFO_1 \ldots \runFIFO_\nch$ such that $\runFIFO_\ch = \recaction{\contents_\ch}{\ch}$ for all $\ch \in \Ch{\nch}$, where $\recaction{\contents_\ch}{\ch}$ is to be understood as a sequence of transitions whose effect is to consume the string $\contents_\ch$ from the channel $\ch$. By construction, we have, $ (\init,\initcontents) \xrightarrow{ \runFIFO\cdot \runFIFO' }_{\fifo_q} (q, \initcontents)$.  Furthermore, $\projsend{\runFIFO}{\ch} = \projsend{\runFIFO\cdot\runFIFO'}{\ch}$ for all $\ch \in \Ch{\nch}$. Hence, $\runFIFO\cdot \runFIFO' \in  \Pref(\sendL)$ and we can conclude that $(q, \contents) \in \reachsetL{\fifo}{\boundedL_!}$ iff $(q, \initcontents) \in \reachsetL{\fifo_{q}}{\boundedL_!}$.

	Therefore, \IBounded control-state reachability reduces to \IBounded reachability for FIFO machines.

	\subsection*{Part (b):}
	
	Given a FIFO machine $\fifo = (\fifostates, \Ch{\nch}, \alphabet, \fifotransition, \init)$, a configuration $(q, \contents)$, and a tuple $\boundedL = (L_\ch)_{\ch \in \Ch{\nch}}$ of non-empty regular bounded languages $L_\ch \subseteq \Sigma_\ch^\ast$, we construct $\fifo_{(q,\contents)}$ as in the case of reducing reachability to control-state reachability (see proof of Proposition~\ref{prop:reach-equiv-csr}(a)). We then modify $\fifo_{(q,w)}$ to $\fifo'$ as follows. We add a new channel $\clive$ to the existing set of channels $\Ch{\nch}$ (the set of channels is now $\Ch{\nch}'$). For all $q \neq \qdead$, we add the following transition:  $(q, \sendaction{\$}{\clive}, q)$. Hence, except for $\qdead$, every control state has at least one send action. Finally, we also construct a new tuple $\boundedLp = (L'_\ch)_{\ch \in \Ch{\nch}'}$ such that $L'_\ch = L_\ch.\$$ for all $\ch \in \Ch{\nch}$ and $L'_\clive = \$^\ast$.

We claim that $(q, \contents) \in \reachsetL{\fifo}{\sendL}$ iff $\reachsetL{\fifo'}{\sendL'}$ contains a deadlock.
	To see this, first, we observe that, if there is a deadlock in $\reachsetL{\fifo'}{\sendL'}$, then the associated control state would be $\qdead$ since if we are in any configuration $s'$ such that the associated control state $q'$ is not $\qdead$, the transition  $(q', \sendaction{\$}{\clive}, q')$ can always be taken, and hence, there will never be a deadlock. 
	
	Let us now suppose that the configuration $(q, \contents)$ is in $\reachsetL{\fifo}{\sendL}$. We can execute the same set of transitions as in $\fifo$ and reach the control state $q$ with the channel contents $\contents$ via an execution $\runFIFO'$ in $\fifo'$. Having done that, we can then execute the path $q \xrightarrow{\runFIFO} \qdead$ as described in $\fifotransition'$ in order to reach $\qdead$. Also observe that $\runFIFO' \cdot \runFIFO \in \sendL'$. Since there are no transitions from this control state, we reach a deadlock. 
	
	Suppose now that $(q, \contents)$ is not in $\reachsetL{\fifo}{\sendL}$. Hence, we cannot reach $(q, \contents)$ in $\fifo'$ and thus, cannot execute $q \xrightarrow{\runFIFO} \qdead$. Furthermore, as we saw previously, we can never be in a deadlock as we can always send $\$$ to the channel $\clive$. 
\end{proof}

\medskip

If, for a given $q \in \fifostates$, we set $\Rel$ to be the universal relation $\ChContents$, \IBounded  rational-reachability captures \IBounded control-state reachability, and if we set $\Rel = \{\contents\}$, we can decide if the configuration $(q, \contents)$ is reachable. In order to reduce \IBounded deadlock to \IBounded rational-reachability, we first explore the control states in order to find the set of states $\fifostates' \subseteq \fifostates$ which allow only receptions (no control states with sends can be part of a deadlock). This set can easily be computed from the set of transitions of the machine. Then, for each $q \in \fifostates'$, we can see if there exists a reachable configuration $(q, \contents)$ such that, for all $\ch$, we have $\contents_\ch \in K_\ch = \{\varepsilon\} \cup \{a.u \mid u \in \Sigma_\ch^\ast \text{ and }a \in \Sigma_\ch$ such that
there is no transition $ (q, \recaction{a}{\ch}, q')$ in $\fifo\}$.
Note that $\Rel_q = \prod_{\ch \in \Ch{\nch}} K_\ch$ is recognizable and, therefore, rational. Furthermore, if there exists such a reachable $(q, \contents)$ with $\contents \in \Rel_q$, then it is a deadlock.
Hence, using the fact that generalized \IBounded  rational-reachability is decidable (Theorem \ref{thm:general-I-bounded-reach}), we immediately deduce the following corollary:

\begin{cor}
	The problems \IBounded reachability, \IBounded control-state reachability, and \IBounded deadlock are decidable for FIFO machines.
\end{cor}

\begin{rem}
	Since input-bounded FIFO machines subsume VASS (a VASS can be seen as an input-bounded FIFO machine with an alphabet reduced to a unique letter), the complexity of \IBounded reachability is not elementary, which is inherited from the lower bound for VASS \cite{CzerwinskiLLLM19}.
	
\end{rem}

\section{Unboundedness and Termination} \label{sec:unboundterm}

We now introduce two new decision problems, as follows.

\begin{defi}[IB finiteness problems]
	Given a FIFO machine $\fifo = (\fifostates, \Ch{\nch}, \alphabet, \fifotransition, \init)$, a control-state $q \in Q$, a configuration $s \in \tscomp{\tsstates}{\fifo}$, and a tuple
	$\boundedL = (L_\ch)_{\ch \in \Ch{\nch}}$ of non-empty regular bounded languages $L_\ch \subseteq \Sigma_\ch^\ast$.
	\begin{itemize}
		\item \emph{\IBounded unboundedness}: Is $\reachsetL{\fifo}{\Pref(\sendL)}$ infinite?
		\item \emph{\IBounded termination}: Is there no infinite execution of the form $\tsinit_\fifo \tstransp{\smash{\beta_1}}_M s_1 \tstransp{\smash{\beta_2}}_M s_2 \tstransp{\smash{\beta_3}}_M \ldots$ such that, for all $i \in \Nat$, we have $s_i \in \tscomp{\tsstates}{\fifo}$, $\beta_i \in \fifoAlpha{\fifo}$, and $\beta_1\ldots \beta_i \in \Pref(\sendL)$?
	\end{itemize}
	
\end{defi}

Decidability of the unboundedness and termination problems can be deduced from the reachability decision problems. We detail the construction in the sections below.

\subsection*{Unboundedness}\label{finiteness-cmzt}

\IBounded unboundedness in FIFO machines reduces to an equivalent problem in counter machines. Given a FIFO machine $\hat\fifo$ and $\hat\boundedL$, the associated FIFO machine $\fifo$ in normal form (with the corresponding tuple $\boundedL$ of
distinct-letter languages), as well as the associated counter machine $\CS$, the following result can be derived.
\begin{prop} \label{prop:boundedness-cm-fifo}
	$\reachsetL{\hat\fifo}{\hat\sendL}$ is infinite iff $\reachsetL{\fifo}{\sendL}$ is infinite iff $\reachsetL{\CS}{\BCounters{\CS} \cap \crun{\Valid}}$ is infinite.
\end{prop}
\begin{proof}
	
	We first show that unboundedness is preserved by the normal-form construction. 
	This essentially follows from Lemma~\ref{lem:normal-form} and the fact that $h$ is length-preserving.
	
	If $(q,\hat\contents) \in \reachsetL{\hat\fifo}{\hat\sendL}$,
	then $((q,q_\A),\contents) \in \reachsetL{\fifo}{\sendL}$ for some
	$q_\A$ and $\contents$ such that $h(\contents) = \hat\contents$.
	Thus, if $\reachsetL{\hat\fifo}{\hat\sendL}$ is infinite, then so is
	$\reachsetL{\fifo}{\sendL}$.
	
	Conversely, if $((q,q_\A),\contents) \in \reachsetL{\fifo}{\sendL}$,
	then $(q,h(\contents)) \in \reachsetL{\hat\fifo}{\hat\sendL}$.
	Thus, if $\reachsetL{\fifo}{\sendL}$ is infinite, then so is
	$\reachsetL{\hat\fifo}{\hat\sendL}$.

	\medskip
	
	Now, let us assume that 	$\reachsetL{\fifo}{\sendL}$ is infinite. Hence, there are infinitely many configurations $(q, \contents) \in \tscomp{\tsstates}{\fifo}$ which are reachable from $(\init, \initcontents)$. From Corollary \ref{reachprop}, for each of these configurations, $(q, \cvalue{\contents}) \in
	\textstyle \reachsetL{\CS}{L}$, where $L = \BCounters{\CS} \cap \crun{\sendL \cap \Pref(\recL)
		\cap \bigcup_{\veca \in \pairs{\contents}} \abL{\veca}}$.
	Since $L \subseteq \BCounters{\CS} \cap \crun{\sendL \cap \Pref(\recL)}$, we have $(q, \cvalue{\contents}) \in \reachsetL{\CS}{\BCounters{\CS} \cap \crun{\sendL \cap \Pref(\recL)}}$. Furthermore, there are only finitely many configurations $(q, \contents) \in \tscomp{\tsstates}{\fifo}$ that correspond to a configuration $(q, \cvalue{\contents}) \in \tscomp{\tsstates}{\CS}$. Hence, if $\reachsetL{\fifo}{\sendL}$ is infinite, $\reachsetL{\CS}{\BCounters{\CS} \cap \crun{\sendL \cap \Pref(\recL)}}$ is infinite.
	
	For the converse direction, let us assume that $\reachsetL{\CS}{\BCounters{\CS} \cap \crun{\sendL \cap \Pref(\recL)}}$ is infinite. Hence, there are infinitely many configurations $(q, \cscontents)$ reachable from $(\init, \initcntcontents)$ via some $\runCS$ such that $\runCS \in \BCounters{\CS} \cap \crun{\sendL \cap \Pref(\recL)}$. Pick one such $(q, \cscontents)$ and $\runCS$. There exists a unique $\runFIFO$ such that $\crun{\runFIFO} = \runCS$. Consider the vector $\veca$ such that $\runFIFO \in \abL{\veca}{}$ and let $\contents = \fvalue{\cscontents}{\veca}{}$. Thus, $\cvalue{\contents} = \cscontents$. Hence, $(q, \cvalue{\contents}) \in
	\textstyle \reachsetL{\CS}{\BCounters{\CS} \cap \crun{\sendL \cap \Pref(\recL)
			\cap \bigcup_{\veca \in \pairs{\contents}} \abL{\veca}}}$. From Corollary~\ref{reachprop}, we can deduce that $(q, \contents) \in \reachsetL{\fifo}{\sendL}$. Therefore, for every $\cscontents$ such that $(q, \cscontents)$ is reachable, there is a corresponding configuration $(q, \contents)$ reachable in $\reachsetL{\fifo}{\sendL}$. 
\end{proof}

In particular, this statement applies to prefix-closed languages, i.e., we can reduce checking whether $\reachsetL{\hat\fifo}{\Pref(\hat\sendL)}$ is infinite to checking whether $\reachsetL{\CS}{\BCounters{\CS} \cap \crun{\Valid}}$ is infinite for some $\CS$ and prefix-closed language $\Valid$. The latter is decidable as we establish in the following. Recall that, by the construction of $\CS$, we then have $\reachset{\CS} = \reachsetL{\CS}{\BCounters{\CS} \cap \crun{\Pref(\Valid)}} = \reachsetL{\CS}{\BCounters{\CS} \cap \crun{\Valid}}$.
%}

The main idea that follows is the reduction of unboundedness of the counter machine to reachability in a modified counter machine. It is not immediate that we can use the results in the literature (for example  \cite{DBLP:conf/fsttcs/DufourdF97}) which reduce boundedness to reachability in Petri nets/VASS\@. This is because the property of monotonicity does not extend to zero tests; if one can execute a zero test at $(q, \cscontents)$, it is not necessarily the case that it can be executed at $(q, \cscontents')$ with $\cscontents \leq \cscontents'$, where we let $\chcontents \leq \chcontents'$ if $\chcontents_x \leq \chcontents'_x$ for all $x \in \Cnt$. However, we show that, for the counter machine that we construct from the FIFO machine, this property does hold. If we are able to show this, the constructions used in the case of VASS can be adapted.

\begin{lem}\label{prop:vass-cm}
	For every execution in $\CS$ of the form $(\init, \initcntcontents) \xrightarrow{\sigma} (q, \chcontents) \xrightarrow{\runFIFO'} (q, \chcontents')$ such that $\cscontents \leq \cscontents'$, the following holds: The only counters that are tested to zero during $\runFIFO'$ already evaluate to zero at $(q, \cscontents)$, and do not change their value throughout the execution of $\runFIFO'$.
\end{lem}
\begin{proof}
	
	Let us assume to the contrary that there is at least one counter which is incremented or decremented during $\runFIFO'$ and also tested to zero during the execution. Without loss of generality, let us consider $\cnt{\ch}{i}$ to be the first counter along the execution $\runFIFO'$ that is tested to zero during $\runFIFO'$ and also incremented/decremented before it was tested to zero.
	\begin{itemize}
		\item Case (1): It has a non-zero value at $(q, \cscontents)$, and is then either decremented, or first incremented and then decremented, and finally tested to zero. Since we know that $\sigma.\sigma' \in \BCounters{\CS}$, no counter tested to zero can then be incremented. Hence, its value will remain zero. But this is a contradiction to our assumption that $\chcontents \leq \chcontents'$. Hence, all the counters with non-zero values at $(q, \cscontents)$ cannot be tested to zero during $\runFIFO'$.
		\item Case (2): It has value zero in $(q, \cscontents)$, and is incremented, then decremented, then tested to zero during $\runFIFO'$. This implies that it first has to be incremented.
		Consider now some
		sub-execution $\runFIFO'' \in \Pref(\runFIFO')$ where $(q, \cscontents) \xrightarrow{\runFIFO''} (q_1, \cscontents_1)$ such that the value of $\cnt{\ch}{i}$ in the configuration $(q_1, \cscontents_1)$ is non-zero. Since there are no ``new'' zero-tests along the execution $\runFIFO''$ (by our assumption),
		we can execute $\runFIFO''$ from $(q, \cscontents')$ (by the monotonicity and trace property). However, we cannot increment the counter $\cnt{\ch}{i}$ along $\runFIFO''$, because it was tested to zero during the run $(\init, \initcntcontents) \mathrel{\smash{\xrightarrow{\runFIFO.\runFIFO'}}} (q, \cscontents')$. Hence, we once again have a contradiction. \qedhere
	\end{itemize}
\end{proof}

Now, we can use results from \cite{DBLP:journals/tcs/FinkelS01} to show the following:

\begin{prop}\label{prop:unbounded-vass}
	The set $\reachset{\CS}$ is infinite iff there exist $\runFIFO, \runFIFO' \in \fifoAlpha{\CS}^\ast$, $q \in Q$, and $\cscontents, \cscontents' \in \N^\Cnt$ such that $(\init, \initcntcontents) \xrightarrow{\sigma} (q, \chcontents) \xrightarrow{\runFIFO'} (q, \chcontents')$ and $\chcontents < \chcontents'$ (i.e., $\chcontents \le \chcontents'$ and $\chcontents \neq \chcontents'$).
\end{prop}

\begin{proof} \label{app:unbVASS}
	Let us assume that $\reachset{\CS}$ is infinite. 
	As in \cite{DBLP:journals/tcs/FinkelS01}, we consider the tree of all prefixes of computations. We prune this tree by removing all prefixes where there is at least a loop, i.e., containing two nodes that are labeled by the same configuration. 
	Since every reachable configuration can be reached without a loop, we still have an infinite number of prefixes in the pruned tree. For every configuration $(q, \chcontents)$ in the tree, there are finitely many successors (as the transition system is finitely branching). Hence, in order for the tree to be infinite, there is at least one infinitely long execution (by K{\"o}nig's lemma). This execution has no loop. Therefore, by Dickson's lemma, there is an infinite subsequence of configurations $(q_1, \chcontents_1)$, $(q_2, \chcontents_2), \ldots$ such that $\chcontents_1 < \chcontents_2 < \ldots$. 
	Once we extract this sequence, since there are only finitely many control states in $Q$, we know that there is at least one pair $(q, \chcontents), (q, \chcontents')$ such that $(\init, \initcntcontents) \xrightarrow{\sigma} (q, \chcontents) \xrightarrow{\sigma'} (q, \chcontents')$ and $\chcontents < \chcontents'$.
	
	Conversely, let us assume that there is an execution $(\init, \initcntcontents) \xrightarrow{\sigma} (q, \chcontents) \xrightarrow{\sigma'} (q, \chcontents')$ such that $\chcontents < \chcontents'$. 
	We know from Lemma~\ref{prop:vass-cm} that the only counters which may be tested for zero during $\runFIFO'$ already evaluate to zero at $(q, \cscontents)$, and do not change their value throughout the execution of $\runFIFO'$. Therefore, all the transitions of the counter system in $\sigma'$ can be considered as VASS operations. Hence, the property of monotonicity holds, i.e., if $(q, \cscontents) \xrightarrow{\sigma'} (q, \cscontents')$ and $(q, \cscontents) < (q, \cscontents')$, then we know there exists an infinite sequence $\chcontents_1 < \chcontents_2 < \ldots$ reachable from configuration $(q, \cscontents')$. 
\end{proof}

\paragraph*{Construction of modified counter system.} 
We modify the counter system $\CS$ and construct a new counter system $\CSmod$ such that $\reachset{\CS}$ is infinite iff a specific configuration is reachable in $\CSmod$. The construction is loosely based on the reduction of boundedness to reachability for Petri Nets in \cite{DBLP:conf/fsttcs/DufourdF97}.

\begin{prop}\label{countermod-boundedness}
We can effectively construct a counter machine $\CSmod$ (with bounded zero tests) and a finite set of configurations $R \subseteq S_{\CSmod}$ such that
	$\reachset{\CS}$ is infinite iff some configuration from $R$ is reachable in $\CSmod$.
\end{prop}

\begin{proof}
We modify the counter machine $\CS$ and construct a new counter machine $\CSmod$ such that $\reachset{\CS}$ is infinite iff a configuration belonging to a finite set is reachable in $\CSmod$. The construction is loosely based on the reduction of boundedness to reachability for Petri Nets in \cite{DBLP:conf/fsttcs/DufourdF97}. Since we do not know the values of $\chcontents$ and $\chcontents'$ a priori, we will try to characterize the general condition. The difference $\chcontents' - \chcontents$
is a non negative vector, with at least one strictly positive component.
We add a duplicate set of counters for every counter in the system. The intuition is that the counter machine non-deterministically moves from operating on both sets to a configuration from where it only operates on this second set. The first set will remain unchanged (with the value $\cscontents$), and the second set will keep track of the values (until it reaches $\cscontents'$). From this configuration (which represents $(q, \cscontents')$), we move to a new control state, $q_\textsf{reach}$. Here, we check for the condition $\chcontents' - \chcontents > \initcntcontents$ by first decrementing each counter in the first set which has a non-zero value in tandem with the corresponding counter in the second set. We do this until all the counters in the first set are equal to zero. If $\chcontents' - \chcontents > \initcntcontents$, then there is at least one counter in the second set with a non-zero counter value. We non-deterministically decrement all the counters in the second set until we reach a configuration that has some counter $c$ in the second set with a value of $1$, and all other counters evaluate to zero. Since there are finitely many such configurations, we can just check every case.
\end{proof}

From the above results, we have the following theorem.

\begin{thm}
	The \IBounded unboundedness problem is decidable for FIFO machines.
\end{thm}

\begin{rem}
	Gouda et al, stated that unboundedness is in EXPSPACE for letter-bounded systems \cite{gouda1987deadlock}. However, they only give an idea of the proof, stating that it can be done in a similar fashion as for the deadlock problem. In the construction for solving the deadlock problem, they reduce the input language to \emph{tally} letter-bounded languages (tally means that the input-language is included in $a^*$ where $a$ is a letter). They add as many channels as letters in the original letter-bounded-language. Furthermore, in order to ensure that no channel is non-empty before the next channel is read, they ensure that in all control states where a later channel is being read, there are reception transitions of previous channel contents which lead to a sink state (where there is never a deadlock). Notice that it is still possible to leave a channel non-empty before the next channel is read. But one never reaches a deadlock in such an ``incorrect'' run, since there is always the option of reading the unread channel contents of the previous channels and reach the sink state. 
	
	However, when we consider this model for unboundedness, there may exist unbounded ``incorrect'' runs since we can leave a channel non-empty and proceed to the next and may have an unbounded run there. Hence, it seems that one still needs some reachability test to check if the runs are correct because we cannot ensure that some channels are zero in an unbounded run.
\end{rem}

\subsection*{Termination}

For termination, we take a similar approach as for unboundedness.
Suppose again that we are given $\hat\fifo$ and $\hat\boundedL$,
the associated FIFO machine $\fifo$ in normal form (with the corresponding tuple $\boundedL$ of
distinct-letter languages), and the associated counter machine $\CS$.
We first show that the normal form preserves the (non-)termination property.
In the following, $\beta_i$ will denote a send or receive action and
$\alpha_i$ will denote an increment or decrement action.

We obtain the following equivalence.

\begin{prop}
	\label{prop:normal-form-term}
	The following statements are equivalent:
\begin{itemize}\itemsep=1ex
\item There is an infinite execution of the form $\tsinit_{\hat\fifo} \tstransp{\smash{\beta_1}}_{\hat\fifo} s_1 \tstransp{\smash{\beta_2}}_{\hat\fifo} s_2 \tstransp{\smash{\beta_3}}_{\hat\fifo} \ldots$ such that, for all $i \in \Nat$, we have $\beta_1\ldots \beta_i \in \Pref(\hat\sendL)$.
	
\item There is an infinite execution of the form $\tsinit_\fifo \tstransp{\smash{\beta_1}}_\fifo s_1 \tstransp{\smash{\beta_2}}_\fifo s_2 \tstransp{\smash{\beta_3}}_\fifo \ldots$ such that, for all $i \in \Nat$, we have $\beta_1\ldots \beta_i \in \Pref(\sendL)$.
\end{itemize}
\end{prop}

\begin{proof}
Assume that there is an infinite execution of the form $\tsinit_{\hat\fifo} \tstransp{\smash{\hat\beta_1}}_{\hat\fifo} \hat s_1 \tstransp{\smash{\hat\beta_2}}_{\hat\fifo} \hat s_2 \tstransp{\smash{\hat\beta_3}}_{\hat\fifo} \ldots$ in $\hat{\fifo}$ such that, for all $i \in \Nat$, we have $\hat\beta_1\ldots \hat\beta_i \in \Pref(\hat\sendL)$. Hence,
there are $\beta_1,\beta_2,\beta_3, \ldots \in A_\fifo$ such that,
for all $i \in \Nat$, letting $\runFIFO_i = \beta_1 \ldots \beta_i$, we have $h(\runFIFO_i) =  \hat\beta_1\ldots \hat\beta_i $ and $\runFIFO_i \in \Pref(\Valid)$. Hence, we know that, in the FIFO machine $h^{-1}({\hat\fifo})$, one has $(\hat{q}_0, \initcontents)
\xrightarrow{\smash{\beta_1}} (\hat q_1,\contents_1)
\xrightarrow{\smash{\beta_2}} (\hat q_2,\contents_2)
\xrightarrow{\smash{\beta_3}} \ldots$
for suitable $(\hat q_i,\contents_i)$ (by construction of $h^{-1}({\hat\fifotransition})$), and that $\runFIFO_i \in \Pref(L(\A))$.  Therefore, since $\fifo$ is a product of the two machines, we can deduce that there is an infinite execution $\tsinit_\fifo \tstransp{\smash{\beta_1}}_\fifo s_1 \tstransp{\smash{\beta_2}}_\fifo s_2 \tstransp{\smash{\beta_3}}_\fifo \ldots$ as required.

\smallskip

Conversely, assume that there is an infinite execution of the form $\tsinit_\fifo \tstransp{\smash{\beta_1}}_\fifo s_1 \tstransp{\smash{\beta_2}}_\fifo s_2 \tstransp{\smash{\beta_3}}_\fifo \ldots$ such that, for all $i \in \Nat$, we have $\beta_1\ldots \beta_i \in \Pref(\sendL)$. 	Let $\hat{\runFIFO_i} = h(\beta_1\ldots \beta_i)$ for all $i \in \Nat$.
Since $\beta_1\ldots \beta_i  \in \Pref(\sendL)$, we have
$\projsend{\beta_1\ldots \beta_i }{\ch} \in \Pref(L_\ch)$ for all $\ch \in \Ch{\nch}$.
In particular, $\projsend{\beta_1\ldots \beta_i}{\ch} \in h_\ch^{-1}(\Pref(\hat L_\ch))$ and,
therefore, $h_\ch(\projsend{\beta_1\ldots \beta_i}{\ch}) = \projsend{h(\beta_1\ldots \beta_i)}{\ch} \in \Pref(\hat L_\ch)$. We deduce $\hat{\runFIFO_i} \in \Pref(\hat{\sendL})$. 	Furthermore, we can execute $\hat{\runFIFO}$ in $\hat{\fifo}$ (by construction) for all $i \in \Nat$. Hence, we can build an infinite execution in $\hat{M}$, such that $\hat{\runFIFO_i} \in \Pref(\hat\sendL)$ for all $i \in \Nat$.
\end{proof}

The latter property can be reduced to checking a decidable property in the counter machine as follows:

\begin{prop}
	\label{term-cm-term}
	The following statements are equivalent:
\begin{itemize}
\item There is an infinite execution of the form $\tsinit_\fifo \tstransp{\smash{\beta_1}}_\fifo s_1 \tstransp{\smash{\beta_2}}_\fifo s_2 \tstransp{\smash{\beta_3}}_\fifo \ldots$ such that, for all $i \in \Nat$, we have $\beta_1\ldots \beta_i \in \Pref(\sendL)$.
	
\item There is an infinite execution of the form $\tsinit_\CS \tstransp{\smash{\alpha_1}}_\CS s_1 \tstransp{\smash{\alpha_2}}_\CS s_2 \tstransp{\smash{\alpha_3}}_\CS \ldots$ such that, for all $i \in \Nat$, we have $\alpha_1\ldots \alpha_i \in \BCounters{\CS} \cap \crun{\Pref(\Valid)}$.

\item There exist $\runFIFO \in \fifoAlpha{\CS}^\ast$, $\runFIFO' \in \fifoAlpha{\CS}^+$, $(q, \cscontents) \in  \tscomp{\tsstates}{\CS}$, and $\cscontents'$ such that $(\init, \initcntcontents) \xrightarrow{\sigma}_\CS (q, \chcontents) \xrightarrow{\runFIFO'}_\CS (q, \chcontents')$ and $\chcontents \leq \chcontents'$.
\end{itemize}
\end{prop}

\begin{proof}
The equivalence of the first two items is obtained as a corollary of
Propositions~\ref{prop:fifo-cm} and \ref{prop:cm-fifo}.
So let us show that the latter two are equivalent.

	Consider a non-terminating execution as described above. By Dickson's lemma, there is an infinite subsequence of configurations $(q_1, \chcontents_1), (q_2, \chcontents_2), \ldots$ such that $\chcontents_1 \leq \chcontents_2 \leq \ldots$. 
	Once we extract this sequence, since there are only finitely many control states in $Q$, we know that there is at least one pair $(q, \chcontents), (q, \chcontents')$ such that $(\init, \initcntcontents) \xrightarrow{\sigma} (q, \chcontents) \xrightarrow{\sigma'} (q, \chcontents')$ and $\chcontents \leq \chcontents'$.
	
	Conversely, assume $(\init, \initcntcontents) \xrightarrow{\sigma} (q, \chcontents) \xrightarrow{\runFIFO'} (q, \chcontents')$ and $\chcontents \leq \chcontents'$. 
	From the same argument as for boundedness, we can deduce that no counter is being tested to zero for the first time in $\sigma'$ and all the counters previously tested stay unchanged. Hence, all the transitions of the counter system in $\sigma'$ can be considered as VASS operations. Therefore,
	we know there exists $\cscontents''$ such that
	$(q, \cscontents') \xrightarrow{\sigma'} (q, \cscontents'')$ and
	$\cscontents' \leq \cscontents''$. Repeating this reasoning, we can build an infinite sequence 
	starting from $(q, \cscontents')$. 
\end{proof}

Finally, we construct a modified counter machine as in the case for boundedness (Proposition~\ref{countermod-boundedness}) and get the following:

\begin{prop}
We can effectively construct a counter machine $\CS''$ (with bounded zero tests) and a configuration $s \in S_{\CSmodp}$ such that the following statements are equivalent:
\begin{itemize}
\item There exist $\runFIFO \in \fifoAlpha{\CS}^\ast$, $\runFIFO' \in \fifoAlpha{\CS}^+$, $(q, \cscontents) \in  \tscomp{\tsstates}{\CS}$, and $\cscontents'$ such that $(\init, \initcntcontents) \xrightarrow{\sigma}_\CS (q, \chcontents) \xrightarrow{\runFIFO'}_\CS (q, \chcontents')$ and $\chcontents \leq \chcontents'$.

\item Configuration $s$ is reachable in $\CS''$.
\end{itemize}
\end{prop}

\begin{proof}

We can adapt the construction of $\CS'$ for unboundedness.
The difference is that we now allow for $\chcontents = \chcontents'$.
Hence, there is no need anymore to check that, after
decrementing both sets of counters in tandem, there is
still a positive counter left in the second set.
We can therefore also empty the second set of counters in a new
control-state $q$ and check whether $(q,\initcntcontents)$ is reachable.
\end{proof}

Thus, we have the following theorem.

\begin{thm}
	The \IBounded termination problem is decidable for FIFO machines.
\end{thm}

\section{Output-Bounded Problems} \label{sec:opbounded}

We consider the dual case of input-bounded languages in which the set of words that may be received by each channel (the output-language) is constrained to be bounded. The \emph{\OBounded problems} are defined as follows:

\begin{defi}[\OBounded decision problems]

	Given a FIFO machine $\fifo = (\fifostates, \Ch{\nch}, \alphabet, \fifotransition, \init)$,
a tuple $\boundedL = (L_\ch)_{\ch \in \Ch{\nhc}}$ of non-empty regular bounded languages
$L_\ch \subseteq \Sigma_\ch^\ast$ (each given in terms of a finite automaton),
a control state $q \in Q$, a configuration $s \in \tscomp{\tsstates}{\fifo}$, and a rational relation $\Rel \subseteq \ChContents$.

	\begin{itemize}
		\item \emph{\OBounded rational reachability}: Do we have $(q,\contents) \in \reachsetL{\fifo}{\recL}$ for some $\contents \in \Rel$?
		\item \emph{\OBounded reachability}: Do we have $s \in \reachsetL{\fifo}{\recL}$\,?
		\item \emph{\OBounded control-state reachability}: Do we have
		$(q,\contents) \in \reachsetL{\fifo}{\recL}$ for some $\contents$\,?
		\item \emph{\OBounded deadlock}: Does $\reachsetL{\fifo}{\recL}$ contain a deadlock\,?
		
				\item \emph{\OBounded unboundedness}: Is $\reachsetL{\fifo}{\Pref(\recL)}$ infinite?
				\item \emph{\OBounded termination}: Is there no infinite execution of the form $\tsinit_\fifo \tstransp{\smash{\beta_1}}_\fifo s_1 \tstransp{\smash{\beta_2}}_\fifo s_2 \tstransp{\smash{\beta_3}}_\fifo \ldots$ such that, for all $i \in \Nat$, we have $s_i \in \tscomp{\tsstates}{\fifo}$, $\beta_i \in \fifoAlpha{\fifo}$, and $\beta_1\ldots \beta_i \in \Pref(\recL)$?
	\end{itemize}
\end{defi}

\begin{thm}
	\OBounded reachability is decidable for FIFO machines.
\end{thm}

\begin{proof}
	Given a configuration $(q, \contents)$ in $\fifo$, and a tuple of bounded languages $\boundedL$, the output-bounded reachability problem asks if $(q, \contents) \in \reachsetL{\fifo}{\boundedL_?}$. Since the output language is in $\boundedL_?$, we know that the contents of the channels which have already been read is in the corresponding input-language, i.e., $\boundedL_!$. Therefore, $(q, \contents) \in \reachsetL{\fifo}{\boundedL_?}$ iff $(q, \contents) \in \reachsetL{\fifo}{\boundedLpem}$ where $\boundedLp = (L'_\ch)_{\ch \in \Ch{\nch}}$ and $L'_\ch = L_\ch.\contents_\ch$. Hence, the \OBounded reachability problem is decidable for FIFO machines.
\end{proof}

\begin{thm}\label{thm:OB-controlstate}
	\OBounded control-state reachability is decidable for FIFO machines.
\end{thm}

\begin{proof}
	In order to show that \OBounded control-state reachability is decidable for FIFO machines, we first convert it to the normal form. This construction is similar to that as specified in Section~\ref{sec:bounded-reachability}, however we make a few changes. Firstly, when we change the alphabet to distinct letter, we add an additional letter, say $\$$, to represent all the letters which are not present in the bounded language, but are present in the transitions of the FIFO system. Then, in addition to the transitions we already add to the new FIFO system as in Section~\ref{sec:bounded-reachability}, for every send action in the original FIFO system, we add a send action with this new letter $\$$. However, for the reception, we leave it as before. This new FIFO system can now only read letters in the output language, however, it can potentially send any letter. Hence, we have only restricted the output language.
	
	Next, to the automata that is constructed to accept the bounded language for send actions, we add to all the states a transition that enables the automata to send $\$$ and go to a sink state which loops with send actions sending $\$$. Hence, once again, we have ensured that the input language is not restricted. To the reception automata, we make no such changes. Hence, the reception automata only accept the bounded language.  
	
	 In this new machine, we can solve for control state reachability, over the input-language $\boundedLp = (L'_\ch)_{\ch \in \Ch{\nch}}$ such that $L'_\ch = L_\ch.\$^\ast$ (but we let the output-language remain $\boundedL$ when we construct the trimmed automata). If a configuration $(q, \contents)$ is reachable in this new machine, then there exists a path which can be taken in the original machine to reach control state $q$ via $\sigma \in \recL$. Furthermore, if there is a path which can be taken in the original machine, it can also be taken in the new machine. Hence, we can decide if there exists a $\contents$ such that $(q, \contents)$ is reachable. 
\end{proof}

\begin{thm}
	\OBounded unboundedness and \OBounded termination problems are decidable for FIFO machines.
\end{thm}

\begin{proof}
		We can decide the \OBounded unboundedness and termination problems as well. We can construct a  counter system for the normal form described in the proof of Theorem~\ref{thm:OB-controlstate}. Note that the counter corresponding to $\$$ will only have increments and no decrements, which is in line with the fact that  the contents corresponding to this counter do not belong to the output-language. The FIFO machine has an infinite run iff this newly constructed counter machine has also an infinite run. We then construct the modified counter machine, as is the case for boundedness (see Section~\ref{finiteness-cmzt}), and test if there is a run $(\init, \initcntcontents) \xrightarrow{\sigma} (q, \chcontents) \xrightarrow{\sigma'} (q, \chcontents')$ such that $\chcontents < \chcontents'$. Since this modified counter system has bounded zero tests (the added counter has no decrements or zero tests associated to it), we can decide the reachability of the configuration, and hence, decide if the FIFO machine is unbounded. A similar explanation can be made for termination.
	\end{proof}

\begin{rem}
The rational reachability problem cannot be directly reduced to the input-bounded case. For the output bounded case, we do not know precisely the reachability set, since we only restrict the output-language. Hence, in order to check if some $\contents \in \Rel$ is reachable, we need to be able to compute the reachability set. Similarly, unlike the input-bounded case, we cannot determine the deadlock problem a priori, since the deadlock problem is reduced to rational reachability for the input-bounded case. Hence, these two problems have been left open.
	\end{rem}

\section{FIFO Machines with a Single Channel} \label{sec:singlechannel}

When we restrict the communication to a single channel, we obtain better upper bound for reachability.

\paragraph*{Upper bound for reachability: EXPTIME}

We consider the model of ordered multi-pushdown systems, studied in \cite{DBLP:journals/ijfcs/AtigBH17}. We define it using our counter systems. A counter system can be seen as a multi-pushdown system with a unary alphabet. Unary ordered multi-pushdown systems (UOMPDSs) are multi-pushdown systems which impose a total order on the counters and limit decrements to the lowest non-empty counter. We use the following result of reachability of UOMPDS.

\begin{thm}[\cite{DBLP:journals/ijfcs/AtigKS14}, Theorem 13]
The reachability problem in UOMPDS is solvable in EXPTIME.
\end{thm}

Consider the counter machine $\CS$ with the semantic restriction $\BCounters{\CS}$ which can be built from a single channel FIFO machine $\fifo$. We see that a counter $x_i$ can only be decremented if the previous counters $x_j$ such that $j<i$ are equal to zero. Therefore, the counter system $\CS$ is a UOMPDS, where the order of the counters is defined as $x_i < x_j$ iff $i<j$. Hence, we have the following proposition.

\begin{prop}
	\IBounded reachability in single channel FIFO machines is polynomially reducible to reachability in UOMPDS.
\end{prop}

In \cite{DBLP:journals/ijfcs/AtigKS14}, it was also shown that repeated reachability for UOMPDS is solvable in EXPTIME\@. We know that a system is non-terminating if and only if we can reach any control state infinitely often. Therefore, we can guess a control state and verify if it is reachable repeatedly in order to verify if the system is non-terminating. Furthermore, we see that even for FIFO machines with a single channel, using the same construction as in Proposition~\ref{prop:reach-equiv-csr}, reachability and control state reachability are recursively equivalent.

We then have the following corollary.

\begin{cor}
	The  \IBounded reachability, termination and control-state reachability problems are in EXPTIME for FIFO machines with a single channel.
\end{cor}

However, for the unboundedness result in Section~\ref{finiteness-cmzt}, we use the reachability for counter systems with bounded zero tests. Furthermore, the set of all counters in the modified counter system do not have a total order anymore. Hence, following our constructions, we may only say that the unboundedness problem in FIFO machines with a single channel (over a bounded language) is solvable by using the Petri net reachability (which is tower-hard).

\begin{rem}
	In the case of a single channel, we have a total order on the counters, and therefore we are able to use results from Ordered Multipushdown Systems. However, even when we consider two channels, we cannot immediately extend these results since there is no longer a total order in the decrement of counters. 
\end{rem}

\paragraph*{Lower bound for reachability: NP-hard} 

We consider a sub-problem of the IB reachability, unboundedness, and termination problems as follows. Recall that, given a bounded language $L = w_1^\ast \ldots w_n^\ast$, if $|w_1| = \cdots = |w_n| = 1$, i.e., $w_1, \ldots, w_n \in A$, then $L$ is called a letter-bounded language.

\begin{defi}[\ILBounded decision problems]
	Given a FIFO machine $\fifo = (\fifostates, \Ch{\nch}, \alphabet, \fifotransition, \init)$, a control-state $q \in Q$, a configuration $s \in \tscomp{\tsstates}{\fifo}$, and a tuple
	$\boundedL = (L_\ch)_{\ch \in \Ch{\nch}}$ of non-empty regular letter-bounded languages $L_\ch \subseteq \Sigma_\ch^\ast$.
	\begin{itemize}
		\item \emph{\ILBounded reachability}: Do we have $s \in \reachsetL{\fifo}{\sendL}$\,?
		\item \emph{\ILBounded unboundedness}: Is $\reachsetL{\fifo}{\Pref(\sendL)}$ infinite?
				\item \emph{\ILBounded termination}: Is there no infinite execution of the form $\tsinit_\fifo \tstransp{\smash{\beta_1}}_\fifo s_1 \tstransp{\smash{\beta_2}}_\fifo s_2 \tstransp{\smash{\beta_3}}_\fifo \ldots$ such that, for all $i \in \Nat$, we have $s_i \in \tscomp{\tsstates}{\fifo}$, $\beta_i \in \fifoAlpha{\fifo}$, and $\beta_1\ldots \beta_i \in \Pref(\sendL)$?
	\end{itemize}
\end{defi}

We see that most of the properties are $\NP$-hard for 
FIFO machines with a single channel over a letter-bounded language. This is proved by simulating 3-CNF formula with such machines. Our simulation follows the same ideas as the proof of $\NP$-hardness for flat FIFO machines with \emph{multiple} channels \cite{DBLP:conf/lics/EsparzaGM12,DBLP:conf/concur/FinkelP19}, except that we use a unique channel.

\begin{thm}
	\label{prop:NPhard}
	\ILBounded reachability, \ILBounded unboundedness, and \ILBounded non-termination are $\NP$-hard for machines with a single channel, even when the input language is letter bounded.
\end{thm}

\newcommand{\clause}{C}

\begin{proof}
	We reduce from 3SAT\@. Given a 3-CNF formula $\clause_1 \land \ldots \land \clause_m$ over variables $x_1, \ldots , x_n$, we construct a FIFO machine with one channel. The message alphabet has $2n + 1$ letters and is as follows $\Sigma_\# = \Sigma \uplus \{\#\}$, where $\Sigma = \{ \letter{1}, \ldots, \letter{n} \} \uplus \{\notletter{1}, \ldots, \notletter{n}\}$. 
	The FIFO machine consists of the gadgets shown below. The gadget for variable $x_k$ adds either $\letter{k}$ (in the top transition) or $\notletter{k}$ (in the bottom edge) to the channel. At the end of this gadget, the channel will have either $\letter{k}$ or $\notletter{k}$. We will sequentially compose the gadgets for all variables. Starting from the initial control state of the gadget for variable $x_1$, we reach the final control state of the gadget for variable $x_n$, such that for every variable $x_k$ we either have  $\letter{k}$ or $\notletter{k}$ in the channel --- and this determines the truth valuation.
	
	We then add the gadget that adds the stop symbol to the channel, as shown in Figure~\ref{fig:vargadget}.
	\begin{figure}[!h]
		
		\begin{tikzpicture}[->, node distance=1.75cm, auto, thick]
			\node[] (p1) [circle, fill, inner sep = 1.5pt] {};
			\node[] (p2) [right= of p1, circle, fill, inner sep = 1.5pt] {};
			
			\node[] (p3) [right= 3cm of p2, circle, fill, inner sep = 1.5pt] {};
			\node[] (p4) [right= of p2, circle, fill, inner sep = 1.5pt] {};
			
			\path[->]
			
			(p1) edge [bend left] node[] {$!\letter{k}$} (p2)
			(p1) edge [bend right] node[swap] {$!\notletter{k}$} (p2)
			(p4) edge [] node[] {$!\#$} (p3)
			;
			\node [below=1cm, align= center, text width=8cm] at (p1)
			{
				(a) Gadget for variable $x_k$
			};
			
			\node [below=1cm, align=center,text width=8cm] at (p3)
			{
				(b) Gadget for stop marker
			};
		\end{tikzpicture}
		\caption{Gadget for variables\label{fig:vargadget}}
	\end{figure}
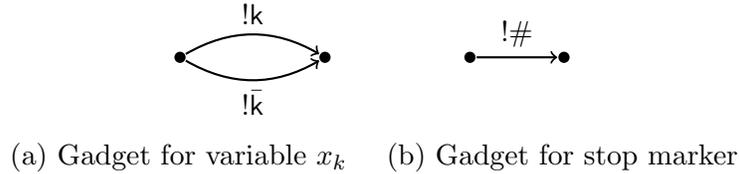
	
	Next, we add gadgets for the clauses. The gadget for the example clause $\clause_1 = x_1 \lor \neg x_2 \lor \neg x_3$ (gadgets for other clauses follow similar pattern) is shown in Figure~\ref{fig:clausegadgets}. The gadget checks that the channel has either $\letter{1}$ (in the top path) or has $\notletter{2}$ (in the middle path) or has $\notletter{3}$ (in the bottom path).  We append the clause gadgets to the end of the variable gadgets one after the other. All clauses are satisfied by the truth valuation determined by the contents of channels $x_1 , \ldots, x_n$ iff we can reach the last control state of the last clause.
	
	The gadget for cleaning up all the variables is shown below (it receives all the letters from the channel). We append the cleanup gadget to the end of the clause gadget for $\clause_m$.	
	
	\begin{figure}[!h]
		
		\begin{tikzpicture}[->, node distance=1.5cm, auto, thick]
			\node[] (p1) [circle, fill, inner sep = 1.5pt] {};
			\node[] (p2) [right= 1.5cm of p1, circle, fill, inner sep = 1.5pt] {};
			\node[] (p3) [above= 1.5cm of p2, circle, fill, inner sep = 1.5pt] {};
			\node[] (p4) [below= 1.5cm of p2, circle, fill, inner sep = 1.5pt] {};
			\node[] (p5) [right= 2cm of p2, circle, fill, inner sep = 1.5pt] {};
			\node[] (p6) [right= 2cm of p3, circle, fill, inner sep = 1.5pt] {};
			\node[] (p7) [right= 2cm of p4, circle, fill, inner sep = 1.5pt] {};
			\node[] (p8) [right= 2cm of p5, circle, fill, inner sep = 1.5pt] {};
			
			\node[] (p9) [right= 3.5cm of p8, circle, fill, inner sep = 2pt] {};
			%\node[] (p1) {$q$};
			%\node[] (p2) [right= of p1]{$q'_{i,1}$};
			%\node[] (p2) [right= of p1,circle,fill,inner sep=1.5pt] {};
			%\node[] (p3) [right= of p2, circle,fill,inner sep=1.5pt] {};
			\path[->]
			%(p1) edge node[] {$!\$$} (p2)
			%(p2) edge node[] {$?w\$$} (p3)
			%(p2) edge node[] {$!\$$} (p3)
			%(p3) edge [loop above] node[swap] {$?a!a$} (p3)
			(p1) edge [bend left] node[] {$\varepsilon$} (p3)
			(p1) edge [bend right] node[] {$\varepsilon$} (p4)
			(p1) edge [] node[] {$\varepsilon$} (p2)
			(p3) edge [loop] node[swap] {$?\letter{a}!\letter{a}$} (p3)
			(p3) edge [] node[] {$?\letter{1}!\letter{1}$} (p6)
			(p6) edge [loop] node[swap] {$?\letter{a}!\letter{a}$} (p6)
			(p6) edge [bend left] node[] {$?\#!\#$} (p8)
			(p2) edge [loop] node[swap] {$?\letter{a}!\letter{a}$} (p2)
			(p2) edge [] node[]  {$?\notletter{2}!\notletter{2}$} (p5)
			(p5) edge [loop] node[swap] {$?\letter{a}!\letter{a}$} (p5)
			(p5) edge [] node[] {$?\#!\#$} (p8)
			(p4) edge [loop] node[swap] {$?\letter{a}!\letter{a}$} (p4)
			(p4) edge [] node[] (q) {$?\notletter{3}!\notletter{3}$} (p7)
			(p7) edge [loop] node[swap] {$?\letter{a}!\letter{a}$} (p7)
			(p7) edge [bend right] node[swap] {$?\#!\#$} (p8)
			
			(p9) edge [loop] node[swap] {$?\Sigma_\#$} (p9)
			;
			\node [below=1.5cm, align=right,text width=7cm] at (q)
			{
				(a) Gadget for clause $\clause_1 = x_1 \lor \neg x_2 \lor \neg x_3$
			};
			\node [below=2.9cm, align=center,text width=6cm] at (p9)
			{
				(b) Gadget for cleanup
			};
		\end{tikzpicture}
		
		\caption {Gadget for clauses (the loop labelled $?\letter{a}!\letter{a}$ represents loops for all $\letter{a} \in \Sigma$)\label{fig:clausegadgets}}
	\end{figure}
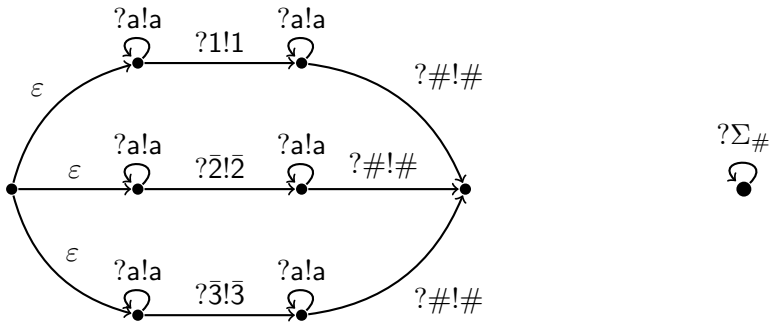
	
	Note that in the FIFO machine given above, every gadget can only be visited once, and the input language of each gadget for a variable $x_k$ is equal to $\{\letter{k}, \notletter{k}\}$ which is included in the letter-bounded language $ \letter{k}^* \notletter{k}^*$. Hence, for every execution $\runFIFO$ along the sequence of variable gadgets, we have $\runFIFO \in \sendL$ where $\boundedL = \letter{1}^*\notletter{1}^*\ldots \letter{n}^* \notletter{n}^*\#^*$. For every clause gadget, we have, once again, an execution $\runFIFO' \in \sendL$. 
	Hence, we see that the input-language of every run can be restricted to
	the bounded language $(\letter{1}^*\notletter{1}^*\ldots \letter{n}^* \notletter{n}^*\#^*)^{m+1}$. 
	Furthermore, while there are loops in the FIFO machine, it can be seen that no loop can be executed infinitely often. We can also see that along every run, the channel is bounded and the size of the channel does not exceed $n+1$.
	
	The given 3-CNF formula is satisfiable iff the control state of the cleanup gadget can be reached with the channel being empty. Hence, this constitutes a reduction to the reachability problem. Furthermore, if we add a self loop to the state of the cleanup gadget, such that it sends the letter $\#$ to the channel, then this loop can be iterated infinitely often to add unboundedly many occurrences of the letter $\#$ to the channel. Now, the given 3-CNF formula is satisfiable iff the constructed 
	FIFO machine is unbounded iff channel is unbounded iff there is a non-terminating run. Hence reachability, unboundedness and non-termination 
	are all NP-hard.
\end{proof}

\begin{rem}
	We can remove the $\varepsilon$-transitions in the above construction by instead non-deterministically choosing one of the three variables per clause. This would eliminate all $\varepsilon$-transitions, and corresponds to the model we define, which does not have them.
	
	We can adapt the proof above to the more restrictive case of FIFO machines whose input language is restricted to a distinct-letter-bounded language, by modifying the transitions in the gadget for clause $\clause_i$ as follows: For every transition sequence of the form $?\letter{a}!\letter{a}$, we replace it by $?\letter{a}_{i-1}!\letter{a}_i$, thereby ensuring that we write different letters to the channel in every gadget. The transitions for the gadgets for each variable $x_k$ (when it is set initially) would be modified by $!\letter{k}_{0}$ and $?\notletter{k}_0$.
\end{rem}

Furthermore, we can modify gadgets for the clauses (see Figure~\ref{flat-gadget}) in order to have the following corollary, which improves a similar result for flat FIFO machines with \emph{multiple} channels in  \cite{DBLP:conf/concur/FinkelP19}.

\begin{cor}\label{flatsingle}
	For flat FIFO machines with a single channel, reachability, unboundedness, and non-termination are NP-hard, hence, they are also NP-complete.
\end{cor}

\begin{figure}[h!]	
	\begin{tikzpicture}[->, node distance=1.72cm, auto, thick, scale=0.90]
		\node[] (p1) [circle, fill, inner sep = 1.5pt] {};
		\node[] (p2) [right= 1.5cm of p1, circle, fill, inner sep = 1.5pt] {};
		\node[] (p2p) [right= 0.75cm of p2, circle, fill, inner sep = 1.5pt] {};
		\node[] (p2p3) [right= 0.75cm of p2p, circle, fill, inner sep = 1.5pt] {};
		\node[] (p3) [above= 1.5cm of p2, circle, fill, inner sep = 1.5pt] {};
		\node[] (p3p) [right= 0.75cm of p3, circle, fill, inner sep = 1.5pt] {};
		\node[] (p3p4) [right= 0.75cm of p3p, circle, fill, inner sep = 1.5pt] {};
		\node[] (p4) [below= 1.5cm of p2, circle, fill, inner sep = 1.5pt] {};
		\node[] (p4p) [right= 0.75cm of p4, circle, fill, inner sep = 1.5pt] {};
		\node[] (p4p5) [right= 0.75cm of p4p, circle, fill, inner sep = 1.5pt] {};
		\node[] (px) [right= 0.75cm of p4p5, circle, fill, inner sep = 1.5pt] {};
		\node[] (p5) [right= 3.5cm of p2, circle, fill, inner sep = 1.5pt] {};
		\node[] (p5p) [right= 0.75cm of p5, circle, fill, inner sep = 1.5pt] {};
		\node[] (p6) [right= 3.5cm of p3, circle, fill, inner sep = 1.5pt] {};
		\node[] (p6p) [right= 0.75cm of p6, circle, fill, inner sep = 1.5pt] {};
		\node[] (p7) [right= 3.5cm of p4, circle, fill, inner sep = 1.5pt] {};
		\node[] (p7p) [right= 0.75cm of p7, circle, fill, inner sep = 1.5pt] {};
		\node[] (p8) [right= 2cm of p5, circle, fill, inner sep = 1.5pt] {};
		
		\node[] (p9) [right= 2.0cm of p8, circle, fill, inner sep = 1.5pt] {};
		\node[] (p10) [right= 0.7cm of p9, circle, fill, inner sep = 1.5pt] {};
		\node[] (p11) [right= 0.7cm of p10, circle, fill, inner sep = 1.5pt] {};
		\node[] (p12) [right= 0.7cm of p11, circle, fill, inner sep = 1.5pt] {};
		\node[] (p13) [right= 0.7cm of p12, circle, fill, inner sep = 1.5pt] {};
		%\node[] (p1) {$q$};
		%\node[] (p2) [right= of p1]{$q'_{i,1}$};
		%\node[] (p2) [right= of p1,circle,fill,inner sep=1.5pt] {};
		%\node[] (p3) [right= of p2, circle,fill,inner sep=1.5pt] {};
		\path[->]
		%(p1) edge node[] {$!\$$} (p2)
		%(p2) edge node[] {$?w\$$} (p3)
		%(p2) edge node[] {$!\$$} (p3)
		%(p3) edge [loop above] node[swap] {$?a!a$} (p3)
		(p1) edge [bend left] node[] {$\varepsilon$} (p3)
		(p1) edge [bend right] node[] {$\varepsilon$} (p4)
		(p1) edge [] node[] {$\varepsilon$} (p2)
		%(p3) edge [loop] node[swap] {$?\letter{a}!\letter{a}$} (p3)
		%(p3) edge [] node[] {$?\letter{1}!\letter{1}$} (p6)
		(p3) edge [] node[] {$?\letter{1}!\letter{1}$} (p3p)
		(p3p) edge [bend left] node[] {$?\letter{2}!\letter{2}$} (p3p4)
		(p3p) edge [bend right] node[swap] {$?\notletter{2}!\notletter{2}$} (p3p4)
		(p3p4) edge [draw=white] node[] {$\ldots$} (p6)
		%(p5) edge [loop] node[swap] {$?\letter{a}!\letter{a}$} (p5)
		(p6) edge [bend left] node[] {$?\letter{n}!\letter{n}$} (p6p)
		(p6) edge [bend right] node[swap] {$?\notletter{n}!\notletter{n}$} (p6p)
		(p6p) edge [bend left] node[] {$?\#!\#$} (p8)
		%(p6) edge [loop] node[swap] {$?\letter{a}!\letter{a}$} (p6)
		%(p6) edge [bend left] node[] {$?\#!\#$} (p8)
		(p2) edge [bend left] node[] {$?\letter{1}!\letter{1}$} (p2p)
		(p2) edge [bend right] node[swap] {$?\notletter{1}!\notletter{1}$} (p2p)
		(p2p) edge [] node[] {$?\notletter{2}!\notletter{2}$} (p2p3)
		(p2p3) edge [draw=white] node[] {$\ldots$} (p5)
		%(p5) edge [loop] node[swap] {$?\letter{a}!\letter{a}$} (p5)
		(p5) edge [bend left] node[] {$?\letter{n}!\letter{n}$} (p5p)
		(p5) edge [bend right] node[swap] {$?\notletter{n}!\notletter{n}$} (p5p)
		(p5p) edge [] node[] {$?\#!\#$} (p8)
		(p4) edge [bend left] node[] {$?\letter{1}!\letter{1}$} (p4p)
		(p4) edge [bend right] node[swap] {$?\notletter{1}!\notletter{1}$} (p4p)
		(p4p) edge [bend left] node[] {$?\letter{2}!\letter{2}$} (p4p5)
		(p4p) edge [bend right] node[swap] {$?\notletter{2}!\notletter{2}$} (p4p5)
		(p4p5) edge [] node[] {$?\notletter{3}!\notletter{3}$} (px)
		(px) edge [draw=white] node[] {$\ldots$} (p7)
		(p7) edge [bend left] node[] {$?\letter{n}!\letter{n}$} (p7p)
		(p7) edge [bend right] node[swap] {$?\notletter{n}!\notletter{n}$} (p7p)
		%(p4) edge [loop] node[swap] {$?\letter{a}!\letter{a}$} (p4)
		%(p4) edge [] node[] {$?\notletter{3}!\notletter{3}$} (p7)
		%(p7) edge [loop] node[swap] {$?\letter{a}!\letter{a}$} (p7)
		(p7p) edge [bend right] node[swap] {$?\#!\#$} (p8)
		
		(p9) edge [bend left] node[] {$?\letter{1}$} (p10)
		(p9) edge [bend right] node[swap] {$\varepsilon$} (p10)
		(p10) edge [bend left] node[] {$?\notletter{1}$} (p11)
		(p10) edge [bend right] node[swap] {$\varepsilon$} (p11)
		(p11) edge [draw=white] node[] {$\ldots$} (p12)
		%(p12) edge [bend left] node[] {$?\notletter{n}$} (p13)
		(p12) edge [] node[] {$?\#$} (p13)
		
		;
		\node [below=1.5cm, align=right,text width=7.5cm] at (p4p5)
		{
			(a) Gadget for clause $\clause_1 = x_1 \lor \neg x_2 \lor \neg x_3$ 
		};
		\node [below=3.2cm, align=center,text width=5cm] at (p11)
		{
			(b) Gadget for cleanup
		};
	\end{tikzpicture}
	
	\caption{Showing NP-hardness for flat systems with a single channel\label{flat-gadget}}
	
\end{figure}
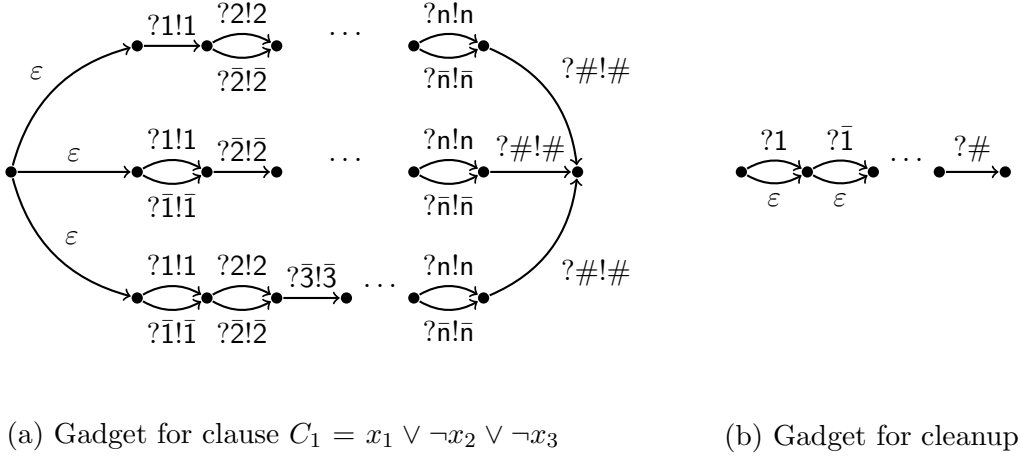

\begin{rem}
	When considering FIFO systems of communicating processes with two FIFO channels (and not FIFO machines with a single channel), the deadlock and the unboundedness are PSPACE for systems of two processes with two one-directional FIFO channels if (at least) one channel is letter-input-bounded \cite{gouda1987deadlock}. We conjecture that this result can be upgraded to input-bounded FIFO systems.
\end{rem}

\section{Conclusion and Perspectives}\label{sec:conclusion}

We extend recent results of the \emph{bounded verification} of
communicating finite-state machines
(equivalently FIFO machines) \cite{DBLP:conf/lics/EsparzaGM12} and of \emph{flat} FIFO machines \cite{DBLP:conf/concur/FinkelP19} by using bounded languages for controlling the input-languages of FIFO channels (and not for controlling the runs of the machine).
We extend old and recent results about input-bounded FIFO machines (see Table~\ref{Table-results}). In particular, we introduce the rational-reachability problem, which subsumes most of the well-known variants of reachability problems like: the (classical) reachability problem, the control-state  reachability problem, and the deadlock problem. We also unify the terminology to facilitate the comparison between results.
Moreover, note that, for most problems (except general/rational reachability), we can reduce \emph{output-bounded} reachability to an equivalent input-bounded problem.
There are still many open problems and challenges:
\begin{itemize}
	\item What is the precise complexity of the five problems for input-bounded FIFO machines with a fixed number of channels?
	\item What is the precise complexity of control-state reachability, deadlock, unboundedness, and termination for input-bounded FIFO machines?
	\item The size of the counter machine associated with a FIFO machine and a tuple of bounded languages is exponential, but only polynomial when we start from a normal form. It will be interesting to see whether the use of existing tools for counter machines is feasible for the verification of FIFO machines from case studies. Case studies shall also reveal how many FIFO machines/systems are actually boundable and/or flattable.
\end{itemize}

\begin{table}[t]
	\caption{Summary of key results; results for all other extensions are subsumed by these results (D stands for decidable).}\label{Table-results}
	\centering
	\begin{tabular}{ *{4}{c} }
		\toprule
		& \multicolumn{1}{c}{\bfseries Flat}
		& \multicolumn{1}{c}{\bfseries Letter-bounded}
		& \multicolumn{1}{c}{\bfseries Bounded}              \\
		\midrule
		UNBOUND			&	NP-C (\cite{DBLP:conf/concur/FinkelP19})	 &	D (\cite{gouda1987deadlock}) 																			&	D (\cite{DBLP:journals/tcs/JeronJ93})			\\
		
		TERM		& NP-C (\cite{DBLP:conf/concur/FinkelP19})	 &   D   											  	     & \textbf{D}   \\
		
		REACH   		& NP-C (\cite{DBLP:conf/concur/FinkelP19})		&  D   												    &  \textbf{D, not ELEM}	\\
		
		CS-REACH   	& NP-C (\cite{DBLP:conf/lics/EsparzaGM12,DBLP:conf/concur/FinkelP19})		&  D   												    &  D	\\
		
		DEADLOCK 	 	& 	D		  & 	D (\cite{gouda1987deadlock})  																 &    \textbf{D}     \\
		\bottomrule
		
	\end{tabular}
\end{table}

\noindent
{\bf Towards a theory of boundable FIFO machines.}~%
In Example~\ref{ex:boundable}, we have seen that all configurations
that are reachable in the CDP protocol are already reachable
in presence of a suitable collection $\boundedL$ of bounded input-languages.
By analogy with the well-established theory of \emph{flattable} machines \cite{BFLP-sttt08,DFGD-jancl10,CFS-atpn2011}, we propose the following definition. 
\begin{defi}
	Let $\fifo$ be a FIFO machine and let $\boundedL$ be a tuple of regular bounded languages. We say that $\fifo$ is \emph{$\boundedL$-boundable} if $\reachset{\fifo} = \reachsetL{\fifo}{\boundedL_!}$. We say that $\fifo$ is \emph{boundable} if there exists a tuple $\boundedL$ of regular bounded languages such that $\fifo$ is \emph{$\boundedL$-boundable}.
\end{defi}
Hence, we deduce that
reachability is decidable for $\boundedL$-boundable FIFO machines, which is a \emph{strictly larger} class than input-bounded machines.
CDP is not input-bounded but it is $\boundedL_{\mathit{CDP}}$-boundable with $\boundedL_{\mathit{CDP}}= ((ab)^*(a + \epsilon)(ab)^*, \msgc^\ast )$. 
Let us also remark that CDP is flattable by using the bounded set of runs
$(!a!b)^*!a!\msgc?\msgc(!a!b)^*+(!a!b)^*$ (where we omit channel information for readability), because it covers the reachability set which is equal to $(ab)^*(a + \epsilon)(ab)^*$ on control-state $(0,0)$.
It is not clear whether reachability is decidable for boundable machines. A strategy that would fairly enumerate \emph{all} regular bounded families $\boundedL_1, \boundedL_2,\ldots,\boundedL_n,\ldots$ will necessarily find the good one, if $\fifo$ is boundable, but this is not sufficient because we must be able to \emph{recognize} $Reach_M$.
Observe that boundable machines are more robust than flat machines. Consider a system $\mathcal{S}=(\A_1,\A_2,\ldots,\A_n)$ of $n$ flat finite automata $\A_i$ communicating peer to peer (P2P) through one-directional FIFO channels. Let $M_\mathcal{S}$ denote FIFO the machine obtained as the Cartesian product of all automata $\A_i$ of $\mathcal{S}$; there is no reason to assume that $M_\mathcal{S}$ is flattable but it is input-bounded and thus $M_\mathcal{S}$ is $\boundedL$-boundable where $\boundedL$ is easily computable from $\mathcal{S}$.

\bibliographystyle{alphaurl}
\bibliography{bibliography}

\end{document}